\documentclass[a4paper]{article}

\advance\hoffset-1in
\advance\voffset-1in
\advance\hoffset2.95mm
\advance\hoffset-2.95mm
\advance\voffset8.8mm
\usepackage[utf8]{inputenc} 
\usepackage[T1]{fontenc}
\usepackage{microtype} 
\usepackage{xcolor}
\usepackage[hidelinks]{hyperref}

\usepackage{mathtools, amsthm, amsfonts, amssymb}
\usepackage{caption}
\usepackage{hyperref}
\usepackage{authblk}

\usepackage{rotating}

\usepackage{mathrsfs}
\usepackage{booktabs} 
\usepackage{commath}

\usepackage{tikz} 
\usetikzlibrary{positioning}
\usetikzlibrary{matrix, arrows}

\usepackage{tikz,tikz-cd}

\theoremstyle{plain}
\newtheorem{theorem}{Theorem}[section]
\newtheorem{proposition}[theorem]{Proposition}
\newtheorem{lemma}[theorem]{Lemma}
\newtheorem{corollary}[theorem]{Corollary}

\theoremstyle{definition}
\newtheorem{definition}[theorem]{Definition}

\newtheorem*{problem*}{Problem}

\newtheorem{example}[theorem]{Example}

\theoremstyle{remark}
\newtheorem{remark}[theorem]{Remark}

\DeclareMathOperator{\per}{per} 

\newcommand{\Oh}{\mathcal{O}} 

\newcommand{\FF}{\mathbb{F}} 
\newcommand{\CC}{\mathbb{C}} 
\newcommand{\NN}{\mathbb{N}} 
\newcommand{\GL}{\mathrm{GL}} 
\newcommand{\ZZ}{\mathbb{Z}} 

\DeclareMathOperator{\diag}{diag} 

\DeclareMathOperator{\N}{\mathrm{N}} 
\DeclareMathOperator{\tr}{tr} 

\newcommand{\class}[1]{\mathbf{#1}} 

\newcommand{\VP}{\class{VP}} 
\newcommand{\VPk}[1]{\VP_{#1}} 
\newcommand{\VNP}{\class{VNP}} 
\newcommand{\VPe}{\class{VP}_{\!\mathrm{e}}} 
\newcommand{\VPs}{\class{VP}_{\!\mathrm{s}}} 
\newcommand{\VNPe}{\class{VNP}_{\!\mathrm{e}}} 
\newcommand{\VNPk}[1]{\class{VNP}_{#1}} 
\newcommand{\VNPs}{\class{VNP}_{\!\mathrm{s}}} 

\DeclareMathOperator{\homog}{H} 

\newcommand{\weakest}{\mathrm{wst}} 
\newcommand{\weak}{\mathrm{w}} 
\newcommand{\gen}{\mathrm{g}} 

\DeclareMathOperator{\size}{size} 

\DeclareMathOperator{\polyf}{poly} 

\newcommand{\defin}[1]{\emph{#1}} 

\newcommand{\HH}{\mathbb{H}} 
\newcommand{\SSS}{\mathbb{S}} 
\newcommand{\LL}{\mathbb{L}} 

\DeclareMathOperator{\characteristic}{char} 

\newcommand{\eps}{\varepsilon} 

\newcommand{\Eq}{\mathrm{Eq}} 

\newcommand{\xvars}{\mathbf{x}} 
\newcommand{\bvars}{\mathbf{b}} 

\newcommand{\Sop}{\mathrm{S}} 

\def\compactminus{\!-\!} 

\tikzset{node distance=1.1cm and 1.1cm, on grid, semithick, shorten <=2pt, shorten >=2pt,
state/.style ={circle,black, fill, inner sep=0pt, minimum width = 0.3 em}}
\tikzset{every node/.style={circle, fill=white, inner sep = 0.1 em}}

\newcommand{\approxcompl}{\underline{L}}
\newcommand{\Fib}{\mathrm{Fib}}

\newcommand{\approxbar}[1]{\overline{#1}} 

\newcommand{\poly}{\mathrm{poly}} 

\newcommand{\polyapproxbar}[1]{\overline{#1}^{\hspace{0.5pt}\raisebox{-1.2pt}{\scriptsize$\smash{\poly}\!$}}} 

\title{On algebraic branching programs of small width\footnote{This work was partially supported by NWO (617.023.116).}}

\author[1]{Karl Bringmann}
\author[1]{Christian Ikenmeyer}
\author[2]{Jeroen Zuiddam}
\affil[1]{\small Max-Planck-Institut für Informatik, Saarland Informatics Campus, Germany\\
  \texttt{kbringma@mpi-inf.mpg.de}, \texttt{cikenmey@mpi-inf.mpg.de}}
\affil[2]{\small Centrum Wiskunde \& Informatica, Amsterdam, Netherlands\\
  \texttt{j.zuiddam@cwi.nl}}

\begin{document}

\maketitle

\begin{abstract}
In 1979 Valiant showed that the complexity class $\VPe$ of families with polynomially bounded formula size is contained in the class $\VPs$
of families that have algebraic branching programs (ABPs) of polynomially bounded size.
Motivated by the problem of separating these classes we study the
topological closure $\approxbar{\VPe}$, i.e.\ the class of polynomials that can be
approximated arbitrarily closely by polynomials in $\VPe$. We describe
$\approxbar{\VPe}$ with a strikingly simple complete polynomial (in characteristic different from 2) whose recursive
definition is similar to the Fibonacci numbers.
Further understanding
this polynomial seems to be a promising route to new formula lower bounds.

Our methods are rooted in the study of ABPs of small constant width.
In 1992 Ben-Or and Cleve showed that formula size is polynomially equivalent to width-3 ABP size.
We extend their result (in characteristic different from 2)
by showing that approximate formula size is polynomially equivalent to approximate width-2 ABP size.
This is surprising because in 2011 Allender and Wang gave explicit polynomials
that cannot be computed by width-2 ABPs at all!
The details of our construction lead to the aforementioned characterization of~$\approxbar{\VPe}$.

As a natural continuation of this work we prove 
that the class $\VNP$ can be described as the class of families that admit a
hypercube summation of polynomially bounded dimension
over a product of polynomially many affine linear forms.
This gives the first separations of algebraic complexity classes from their nondeterministic analogs. 
\end{abstract}

\section{Introduction}
Let $\VPe$ denote the class of families of polynomials with polynomially bounded formula size
and let $\VPs$ denote the class of families of polynomials that can be written as determinants of matrices of polynomially bounded size whose entries are affine linear forms.
In 1979 Valiant~\cite{Val:79b} proved his famous result $\VPe \subseteq \VPs$.
The question whether this inclusion is strict is a long-standing open question in algebraic complexity theory:
Can the determinant polynomial $\det_n := \sum_{\sigma \in S_n} \text{sgn}(\sigma) \prod_{i=1}^n x_{i,\sigma(i)}$ be computed by formulas of polynomially bounded size?
Motivated by this question we study the class $\overline{\VPe}$
of families of polynomials that can be approximated arbitrarily closely by families in $\VPe$ (see Section\nobreakspace \ref {sec:prelim} for a formal definition).
We present a simple description of the closure $\overline{\VPe}$ and of a
$\approxbar{\VPe}$-complete polynomial whose recursive definition is similar to the Fibonacci numbers, given the characteristic is not 2,
see Theorem~\ref{thm:fibon}.

In algebraic complexity theory, the way of showing a complexity lower bound for a problem $f\in V$ for some $\FF$-vector space $V$ most often goes by (implicitly or explicitly)
finding a function $\mathcal F : V \to \FF$ that is zero on all problems of low complexity while at the same time $\mathcal F(f)\neq 0$.
Grochow \cite{grochow2013unifying} gives a long list (e.g., \cite{Nisan1996, Raz:2009:MFP:1502793.1502797, LO:11v3, approaching-the-chasm-at-depth-four, lamare:13, BI:13})
of settings where complexity lower bounds are obtained in this way.
Moreover, he points out that over the complex numbers these functions $\mathcal F$ can be assumed to be continuous (and even to be so-called highest-weight vector polynomials).
If $\class{C}$ and $\class{D}$ are algebraic complexity classes with $\class{C}\subseteq \class{D}$ (for example, $\class{C}=\VPe$ and $\class{D}=\VPs$),
then any separation of algebraic complexity classes $\class{C} \neq \class{D}$ in this continuous manner would automatically imply the stronger statement
$\class{D} \not\subseteq \overline{\class{C}}$.
It is therefore natural to try to prove the separation
$\VPs \not\subseteq \overline{\VPe}$ instead of the slightly weaker $\VPe \neq \VPs$, which provides further motivation
for studying $\approxbar{\VPe}$.
This is exactly analogous to Mulmuley and Sohoni's geometric complexity approach (see e.g.\ \cite{MR1861288, gct2}
and the exposition \cite[Sec.~9]{MR2861717}) where one tries to prove the separation $\VNP \not\subseteq \approxbar{\VPs}$
to attack Valiant's famous $\VPs \neq \VNP$ conjecture \cite{Val:79b}.
Here $\VNP$ is the class of p-definable families, see Section\nobreakspace \ref {sec:prelim} for a precise definition.

\paragraph*{The generalized Fibonacci polynomial}
We prove that the \emph{generalized Fibonacci polynomial}~$F_n$ is $\overline{\VPe}$-complete under p-degenerations,
where $F_n$ is defined via $F_0\coloneqq 1$, $F_1\coloneqq x_1$, $F_n\coloneqq x_n F_{n-1} + F_{n-2}$, see Section\nobreakspace \ref {subsec:completepoly}.
This means that every family $(f_n)$ in $\overline{\VPe}$
can be obtained as the limit of a sequence $f_n = \lim_{j \to \infty}F_{t(n)}(\ell_1(j),\ldots,\ell_{t(n)}(j))$,
where each $\ell_i(j)$ is a variable or constant and $t(n)$ is a polynomially bounded function.
This is arguably the simplest $\overline{\VPe}$-complete polynomial known today.
Prior to our work the simplest $\approxbar{\VPe}$-complete (and $\VPe$-complete) polynomial was the iterated $3 \times 3$ matrix multiplication polynomial \cite{ben1992computing}.
This immediately motivates the definition of \emph{border Fibonacci complexity} $\approxcompl_\Fib(f)$ of a polynomial~$f$,
which is the smallest number $m$ such that $f$ can be obtained as $\lim_{j \to \infty}(F_m(\ell_1(j),\ldots,\ell_m(j)))_j$.
To make the situation more geometric we allow the $\ell_i(j)$ to be arbitrary affine linear forms.
Our results show that border Fibonacci complexity is polynomially equivalent to border formula size.
This insight is quite striking because a result of Allender and Wang \cite{wang} implies that the Fibonacci complexity \emph{without allowing approximations} can be infinite!

A promising path towards proving formula lower bounds, for example for the determinant or the permanent, is to apply to our setting the following standard geometric ideas.
If we take our field to be the complex numbers and fix the number of variables $n$ and the degree $d$,
then the set of homogeneous degree~$d$ polynomials $\CC[x_1, \ldots, x_n]_d$ contains the set
\[
X_m \coloneqq \{ f \in \CC[x_1, \ldots, x_n]_d \mid \approxcompl_\Fib(f) \leq m \}
\]
as an affine subvariety ($X_m$ is the closure of the set of affine projections of $F_m$ intersected with $\CC[x_1, \ldots, x_n]_d$).
Moreover, since we allowed the $\ell_i(j)$ to be affine linear forms, the group $\GL(\CC^n)$ acts canonically on $X_m$,
making $X_m$ an affine $\GL(\CC^n)$-variety.
If we find a polynomial $\mathcal F$ that vanishes identically on $X_m$, then a nonzero evaluation $\mathcal F(f) \neq 0$ implies that $\approxcompl_\Fib(f)>m$.
This approach looks feasible given the very simple structure of the generalized Fibonacci polynomial.
This is emphasized by the fact that the action of $\GL(\CC^n)$ puts a lot of structure on the coordinate ring of $X_m$, see for example \cite{BI:10, BO:11, LO:11v3, BI:13, HIL:13, Gua:15, OS:16}
where the action of the general linear group on the coordinate ring of a variety is used to classify some of its defining equations.

\subsection{Main Results}

\paragraph*{Algebraic Branching Programs (ABPs) of width 2}
Our main objects of study are the following classes of families of polynomials: the class of families of polynomials with polynomially bounded formula size $\VPe$
(fan-in 2 arithmetic formulas that use additions and multiplications as their operations),
its closure $\overline{\VPe}$, and the nondeterministic variant $\VNP$.
We do so by studying algebraic branching programs of small width.
These are defined as follows.
An \defin{algebraic branching program} (ABP) is a directed acyclic graph with a source vertex $s$ and a sink vertex $t$
that has affine linear forms over the base field $\FF$ as edge labels.
Moreover, we require that each vertex is labeled with an integer (its \emph{layer}) and that edges in the ABP only point from vertices in layer $i$ to vertices in layer $i+1$.
The \emph{width} of an ABP is the cardinality of its largest layer.
The \emph{size} of an ABP is the number of its vertices.
The \defin{value} of an ABP is the sum of the values of all $s$-$t$-paths, where the value of an $s$-$t$-path is the product of its edge labels. We say that an ABP \defin{computes} its value.
The class $\VPs$ coincides with the class of families of polynomials that can be computed by ABPs of polynomially bounded size, see e.g.~\cite{ramprasad}.

For this paper we introduce the class $\VPk{k}$, $k \in \NN$, which is defined as the class of families of polynomials computable by width-$k$ ABPs of polynomially bounded size. It is well-known that $\VPk{k}\subseteq \VPe$ for every $k\geq 1$ (see~Proposition\nobreakspace \ref {prop:vpkinvpe}).
In 1992, Ben-Or and Cleve \cite{ben1992computing} showed that $\VPk{k}=\VPe$ for all $k \geq 3$ (we review the proof, see \protect \MakeUppercase {T}heorem\nobreakspace \ref {benorcleve}).
In 2011 Allender and Wang \cite{wang} showed that width-2 ABPs cannot compute every polynomial, so in particular we have a strict inclusion $\VPk{2} \subsetneq \VPk{3}$.
Let the characteristic of the base field $\FF$ be different from 2.
Our first main result (\protect \MakeUppercase {T}heorem\nobreakspace \ref {approxvpe} and \protect \MakeUppercase {C}orollary\nobreakspace \ref {approxcor}) is that the closure of $\VPk{2}$ and the closure of $\VPe$ are equal,
\begin{equation}\label{eq:vp2bar}
\approxbar{\VPk{2}}=\approxbar{\VPe}.
\end{equation}
Interestingly, as a direct corollary of \eqref{eq:vp2bar} and the result of Allender and Wang, the inclusion $\VPk{2} \subsetneq \overline{\VPk{2}}$ is strict.
It is easy to see that $\VPk{1}$ equals $\overline{\VPk{1}}$ (Proposition\nobreakspace \ref {VP1isVP1bar}), so $\VPk{1}$ and~$\VPk{2}$ are examples of quite similar algebraic complexity classes that behave differently under closure.
Most importantly, from the proof of \eqref{eq:vp2bar} we obtain our results about the generalized Fibonacci polynomial that we mentioned before.

\paragraph*{VNP via affine linear forms}
We define the classes $\VNPe$ and $\VNP$ in the natural way.
In 1980, Valiant~\cite{valiant1980reducibility} showed that $\VNPe = \VNP$
and in this paper we will always view $\VNP$ as the nondeterministic analog of $\VPe$.
To $\VPk{1}$ and $\VPk{2}$ we similarly associate nondeterministic analogs $\VNP_1$ and $\VNP_2$  (see Section\nobreakspace \ref {sec:prelim}).
Using interpolation techniques it is possible to deduce $\VNPk{2}=\VNP$ from~\eqref{eq:vp2bar}, provided the field is infinite.
Using more sophisticated techniques we strengthen this result to get our second main result (\protect \MakeUppercase {T}heorem\nobreakspace \ref {VNPmain}):
\begin{equation}\label{eq:vnp1}
\VNPk{1}=\VNP.
\end{equation}
That is, a family $(f_n)$ is contained in $\VNP$ iff $f_n$ can be written as a hypercube summation of polynomially bounded dimension
over a product of polynomially many affine linear forms.
Using \eqref{eq:vnp1} it is then easy to verify that $\VPk{1} \subsetneq \VNP_1$
and using \cite{wang} yields $\VPk{2} \subsetneq \VNP_2$,
which separates complexity classes from their nondeterministic analogs.
Interestingly $\VNP_1 \subsetneq \VNP$ over the field with 2 elements, see Section\nobreakspace \ref {pro:FIIproof}.

\paragraph*{Restricted ABP edge labels}
Several more results on small-width ABPs, approximation closures, and hypercube summations are proved throughout this paper.
For example, in Section\nobreakspace \ref {sec:differentabpedgelabels} we investigate the subtleties of what happens if we restrict the ABP edge labels to simple affine linear forms, or to variables and constants. The precise relations between complexity classes that we obtain are listed in Figure~\ref{fig:overview} in Appendix\nobreakspace \ref {sec:figures}.
As another example, we strengthen \eqref{eq:vnp1} as follows (\protect \MakeUppercase {T}heorem\nobreakspace \ref {thm:secondmainstronger}):
A family $(f_n)$ is contained in $\VNP$ iff $f_n$ can be written as a 
hypercube summation of polynomially bounded dimension
over a product of polynomially many affine linear forms
that use \emph{at most two variables} each.

\subsection{Related work}
In the boolean setting as well as in the algebraic setting finding lower bounds for the formula size of explicit problems is considered a major open problem.
For the boolean setting we refer the reader to the line of papers \cite{Subbo61, Andreev87, IN:93, PZ93, Hastad98, Tal14},
which results in an explicit function with formula size $\Omega(n^3/(\log^2 n \log \log n))$.

In the algebraic setting the smallest formula for the determinant has size $\Oh(n^{\log n})$, which can be deduced from e.g.~\cite{PI:15}.
The best known lower bound on the formula size of $\det_n$ is $\Omega(n^3)$ by \cite{Kal:85}.
That paper also gives a quadratic lower bound for an explicit polynomial
(note that the lower bound for the determinant is not quadratic in the number of variables).

Toda \cite{toda:92} proved that several definitions for the class $\VPs$ are equivalent, see also~\cite{malod2008}.
In particular $\VPs$ is the class of polynomials that can be written as determinants of matrices of polynomially bounded size whose entries are affine linear forms.
Due to its pure mathematical formulation, lower bounds for this \emph{determinantal complexity} attracted the attention of geometers \cite{MR:04,lamare:13,Alper2015}.
Moreover, Mulmuley and Sohoni's geometric complexity approach \cite{MR1861288, gct2} is also mainly focused on lower bounds for the determinantal complexity
and the symmetries of the determinant polynomial play a key role in their work.
Recently \cite{BIP:16} showed that it is not possible to prove superpolynomial lower bounds on the determinantal complexity
using only information about the occurrences/non-occurrences of irreducible representations in the coordinate rings of the orbit closures of the determinant and the (padded) permanent.
This disproves a major conjecture in geometric complexity theory.
The proof in \cite{BIP:16} is fairly general and also holds for lower bounds on the formula size.
Only very recently the formula size analog to determinantal complexity, the \emph{iterated matrix multiplication complexity} was studied from a geometric perspective \cite{Ges:15}.

There is a large number of publications on lower bounds for constant \emph{depth} circuits and formulas (with superconstant fan-in),
see e.g.\ \cite{AV:08, KOIRAN201256, chasm3, Tavenas20152},
which recently led to the celebrated result \cite{approaching-the-chasm-at-depth-four}
that the permanent does not admit size $2^{o(\sqrt{m})}$ homogeneous $\Sigma\Pi\Sigma\Pi$ circuits in which the bottom fan-in is bounded by $\sqrt{m}$.
In the light of the previous depth-reduction results this seemed very close to separating $\VP$ from $\VNP$.
Several very recent results \cite{2016arXiv160902103E, FSV:17} indicate that new ideas are needed to separate $\VP$ from $\VNP$.

Ben-Or and Cleve \cite{ben1992computing} proved that a family of polynomials has polynomially bounded formula size if and only if it is computable by width-3 ABPs of polynomial size.
An excellent exposition on the history of small-width computation can be found in \cite{wang},
along with an explicit polynomial that cannot be computed by width-2 ABPs: $x_1 x_2 + x_3 x_4 + \cdots + x_{15} x_{16}$.
Saha, Saptharishi and Saxena \cite[Cor.~14]{SSS:09} showed that $x_1 x_2 + x_3 x_4 + x_5 x_6$ cannot be computed by width-2 ABPs that correspond to the iterated matrix multiplication of upper triangular matrices.

B\"urgisser \cite{MR2097213} studied approximations in the model of general algebraic circuits, finding general upper bounds on the error degree.
For most specific algebraic complexity classes~$\class{C}$ the relation between $\class{C}$ and $\overline{\class{C}}$ has not been an active object of study.
As pointed out recently by Forbes \cite{For:16},
Nisan's result \cite{Nisan:Lower_bounds_non_commutative_compts} implies that $\class{C}=\overline{\class{C}}$ for $\class{C}$ being the class of size-$k$ algebraic branching programs on noncommuting variables.
Recently, a structured study of $\overline{\VP}$ and $\overline{\VPs}$ has been started, see \cite{grochow_et_al:LIPIcs:2016:6314}.
By far the most work in lower bounds for topological approximation algorithms has been done in the area of bilinear complexity,
dating back to \cite{BCRL:79, stra:83-2, Lic:84} and more recently \cite{Lan:05, LO:11v3, HIL:13, Zui:15, LM:16}, to list a few.

\subsection{Paper outline} In Section\nobreakspace \ref {sec:prelim} we introduce in more detail the approximation closure and the nondeterminism closure of a complexity class.
In Section\nobreakspace \ref {sec:firstmainresult} we prove the first main result: border formula size is polynomially equivalent to border width-2 ABP size
and the generalized Fibonacci polynomial is $\approxbar{\VPe}$-complete under p-degenerations.
In Section\nobreakspace \ref {sec:secondmain} we prove the second main result: a new description of $\VNP$ as the nondeterminism closure of families that have polynomial-size width-1 ABPs.
The later sections contain details on how to strengthen the result from Section\nobreakspace \ref {sec:secondmain}
and results on the power of ABPs with restricted edge labels.

\section{Nondeterminism and approximation closure}\label{sec:prelim}
In this section we introduce the approximation closure and the nondeterminism analog of a class. A \defin{family} is a sequence of polynomials $(f_n)_{n\in \NN}$. A \defin{class} is a set of families and will be written in boldface,~$\class{C}$.
For an introduction to the algebraic complexity classes $\VPe$, $\VP$, and $\VNP$ we refer the reader to \cite{burgisser1997algebraic}.
We denote by $\poly(n)$ the set of polynomially bounded functions $\NN\to\NN$. 
We define the norm of a complex multivariate polynomial as the sum of the absolute values of its coefficients.
This defines a topology on the polynomial ring $\CC[x_1,\ldots,x_m]$.
Given a complexity measure $L$, say ABP size or formula size, there is a natural notion of approximate complexity that is called \emph{border complexity}.
Namely, a polynomial $f \in \CC[\xvars]$ has \emph{border complexity}~$\underline L^{\text{top}}$ at most~$c$ if there is a sequence of polynomials $g_1, g_2, \ldots$ in $\CC[\xvars]$ converging to $f$ such that each $g_i$ satisfies $L(g_i) \leq c$.
It turns out that for reasonable classes over the field of complex numbers $\CC$, this \emph{topological} notion of approximation is equivalent to what we call \emph{algebraic} approximation (see e.g.\ \cite{MR2097213}).
Namely, a polynomial $f \in \CC[\xvars]$ satisfies $\underline L(f)^{\text{alg}} \leq c$ iff there are polynomials $f_1, \ldots, f_e \in \CC[\xvars]$ such that the polynomial
\[
h := f + \eps f_1 + \eps^2 f_2 + \cdots + \eps^e f_e \in \CC[\eps, \xvars]
\]
has complexity $L_{\CC(\eps)}(h) \leq c$, where $\eps$ is a formal variable and $L_{\CC(\eps)}(h)$ denotes the complexity of $h$ over the field extension $\CC(\eps)$.
This algebraic notion of approximation makes sense over any base field and we will use it in the statements and proofs of this paper.

\begin{definition}\label{barclass}
Let $\class{C}(\FF)$ be a class over the field $\FF$. We define the \defin{approximation closure}~$\approxbar{\class{C}}(\FF)$ as follows: a family $(f_n)$ over $\FF$ is in $\approxbar{\class{C}}(\FF)$ if there are polynomials $f_{n;i}(\xvars) \in \FF[\xvars]$ and a function $e:\NN\to\NN$
such that the family $(g_n)$ defined by
\[
g_n(\xvars) \coloneqq f_n(\xvars) + \eps f_{n;1}(\xvars) + \eps^2 f_{n;2}(\xvars) + \cdots + \eps^{e(n)} f_{n; e(n)}(\xvars)
\]
is in $\class{C}(\FF(\eps))$. We define the \defin{poly-approximation closure} $\polyapproxbar{\class{C}}(\FF)$ similarly, but with the additional requirement that $e(n) \in \polyf(n)$. We call $e(n)$ the \defin{error degree}.
\end{definition}
Interestingly, for subtle reasons, taking the approximation closure $\class C \mapsto \approxbar{\class C}$ is \emph{not} idempotent in general
and hence not a closure operator, but for reasonable classes (like $\VPk{k}$, $\VPe$, and~$\VP$) it is.

\medskip

One can think of $\VNP$ as a ``nondeterminism closure'' of $\VP$. We want to use the nondeterminism closure for general classes.
\begin{definition}\label{Ndef} Let $\class{\class{C}}$ be a class.
The class~$\N(\class{C})$ consists of families $(f_n)$ with the following property: there is a family $(g_n) \in \class{C}$ and $p(n), q(n) \in \polyf(n)$ such that
\[
f_n(\xvars) = \sum_{\mathclap{\bvars\in \{0,1\}^{p(n)}}} g_{q(n)}(\bvars,\xvars),
\]
where $\xvars$ and $\bvars$ denote sequences of variables $x_1, x_2, \ldots$ and $b_1, b_2, \ldots, b_{p(n)}$.
We will sometimes say that $f(\xvars)$ is a \defin{hypercube sum over $g$} and that $b_1, b_2, \ldots, b_{p(n)}$ are the \defin{hypercube variables}.
For any $s, t$, we will use the standard notation $\VNP_{\!s}^t$ to denote $\N(\VP_{\!s}^t)$, where the superscript~$t$ will become relevant in Section\nobreakspace \ref {sec:differentabpedgelabels}.
We remark that the map $\class C \mapsto \N(\class C)$ trivially satisfies all properties of being a closure operator.
\end{definition}

\section{Approximate width-2 ABPs and formula size}\label{sec:firstmainresult}\label{subsec:completepoly}

As mentioned in the introduction, Allender and Wang \cite{wang} showed that there exist polynomials that cannot be computed by any width-2 ABP, for example
the polynomial $x_1 x_2 + x_3 x_4 + \cdots + x_{15} x_{16}$. Therefore, we have a separation $\VPk{2} \subsetneq \VPk{3} = \VPe$.
We show that allowing approximation changes the situation completely: every polynomial can be approximated by a width-2 ABP.
In fact, every polynomial can be approximated by a width-2 ABP of size polynomial in the formula size, and with error degree polynomial in the formula size.
This is the main result of this section.

\begin{theorem}\label{approxvpe}
$\VPe \subseteq \polyapproxbar{\VPk{2}}$ when $\characteristic(\FF)\neq 2$.
\end{theorem}
We leave as an open question what happens in characteristic 2.

%
In order to understand the following proofs and the corresponding figures 
it is advisable to recall that
an ABP corresponds naturally to an iterated product of matrices if we number the vertices in each layer consecutively, starting with 1.
Namely, consider two consecutive layers $i$ and $i+1$ and let~$M_i$ be the matrix
whose entry at position $(v,w)$ is the label of the edge from vertex $v$ in layer $i$ to vertex $w$ in layer $i+1$
(or 0 if there is no edge between these vertices).
Then the ABP's value equals the product $M_k\cdots M_2 M_1$.

For a polynomial $f$ over $\FF(\varepsilon)$ define the matrix
$
Q(f) \coloneqq
\bigl(\begin{smallmatrix}
f & 1\\
1 & 0
\end{smallmatrix}
\bigr).
$
A \defin{parametrized affine linear form} is an affine linear form over the field $\FF(\eps)$.
A \defin{primitive Q-matrix} is any matrix~$Q(\ell)$, where $\ell$ is a parametrized linear form.
For a $2\times 2$ matrix $M$ with entries in $\FF(\eps)[\xvars]$,  we use the shorthand notation $M + \Oh(\varepsilon^k)$ for
$M + \Big(\begin{smallmatrix}\Oh(\varepsilon^k)&\Oh(\varepsilon^k)\\\Oh(\varepsilon^k)&\Oh(\varepsilon^k)\end{smallmatrix}\Big)$,
where $\Oh(\varepsilon^k)$ denotes the set $\eps^k\, \FF[\eps, \xvars]$.
As a product of matrices, the ABP construction in our proof of \protect \MakeUppercase {T}heorem\nobreakspace \ref {approxvpe} will be of the form $\begin{psmallmatrix}1 & 0\end{psmallmatrix} M_\ell \cdots M_2 M_1 \begin{psmallmatrix}1 \\ 0\end{psmallmatrix}$ where the $M_i$ are primitive Q-matrices $Q(f)$ for which $f$ is either a constant from $\FF(\varepsilon)$ or a variable. We are thus proving a slightly stronger statement than the statement of \protect \MakeUppercase {T}heorem\nobreakspace \ref {approxvpe}.


\begin{lemma}[Addition]\label{lem:addition}
Let $k\geq 1$. Let $f,g \in \FF[\xvars]$ be polynomials such that some $F \in Q(f) + \Oh(\eps^k)$ and some $G \in Q(g) + \Oh(\eps^k)$ can be written as a product of $n$ and $m$ primitive Q-matrices, respectively. Then some matrix  $H \in Q(f+g) + \Oh(\eps^k)$ can be written as the product of $n + m + 1$ primitive Q-matrices. Moreover, if the error degrees in $F, G$ are $e_f, e_g$, respectively, then the error degree of $H$ is at most $e_f + e_g$.
\end{lemma}
\begin{proof}
Note that $(Q(f) + \Oh(\eps^k)) \cdot Q(0) \cdot (Q(g) + \Oh(\eps^k)) = Q(f+g) + \Oh(\eps^k)$, so we have $H\coloneqq F\cdot Q(0)\cdot G \in Q(f+g) + \Oh(\eps^k)$. Moreover, the largest power of $\eps$ occurring in $H$ is $\eps^{e_f + e_g}$. See Fig.\nobreakspace \ref {fig:addition}.
\end{proof}

\begin{figure}[h]
\centering
\begin{minipage}{7em}
\begin {tikzpicture}
\node[state] (11) [label={$u_1$}]{};
\node[state] (12) [right=of 11, label={$u_2$}] {};
\node[state] (21) [below=of 11, label={-90:$v_1$}] {};
\node[state] (22) [below=of 12, label={-90:$v_2$}] {};
\path (11) edge (22);
\path (12) edge (21);
\path (11) edge node[left, xshift=-1pt] {$f + g$} (21);
\path (12) edge[draw=none] node[right] {$\hspace{0.5em}+\, \Oh(\eps^k)$} (22);
\end{tikzpicture}
\end{minipage}
\hspace{5em}$\sim$\hspace{2em}
\begin{minipage}{10em}
\begin {tikzpicture}
\node[state] (11) [label={$u_1$}]{};
\node[state] (12) [right=of 11, label={$u_2$}] {};
\node[state] (21) [below=of 11] {};
\node[state] (22) [below=of 12] {};
\node[state] (31) [below=of 21] {};
\node[state] (32) [below=of 22] {};
\node[state] (41) [below=of 31, label={-90:$v_1$}] {};
\node[state] (42) [below=of 32, label={-90:$v_2$}] {};
\path (11) edge (22);
\path (12) edge (21);
\path (12) edge[draw=none] node[right] {$\hspace{0.5em}+\, \Oh(\eps^k)$} (22);
\path (11) edge node[left, xshift=-1pt] {$g$} (21);
\path (21) edge (32);
\path (22) edge (31);
\path (31) edge (42);
\path (32) edge (41);
\path (32) edge[draw=none] node[right] {$\hspace{0.5em}+\, \Oh(\eps^k)$} (42);
\path (31) edge node[left, xshift=-1pt] {$f$} (41);
\end{tikzpicture}
\end{minipage}
\caption{Addition construction for \protect \MakeUppercase {L}emma\nobreakspace \ref {lem:addition}}\label{fig:addition}
\end{figure}
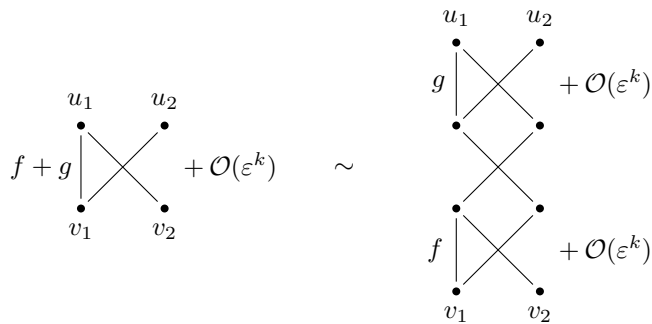

\begin{lemma}[Squaring]\label{lem:squaring}
Let $f\in \FF[\xvars]$ be a polynomial such that some $F \in Q(f) + \Oh(\eps^3)$ can be written as the product of $n$ primitive Q-matrices. Then some matrix $H \in Q(f^2) + \Oh(\eps)$ and some matrix $H' \in Q(-f^2) + \Oh(\eps)$ can be written as the product of $2n + 11 $ primitive Q-matrices. Moreover, if the error degree in $F$ is $e_f$ then the error degree of $H$ and $H'$ is at most $2\cdot e_f + 4$. 
\end{lemma}

\begin{proof}
We set
\begin{align*}
A &\coloneqq \biggl(\begin{matrix} -\eps^{-1} & 0 \\ 0 & \eps \end{matrix}\biggr) = Q(-\eps^{-1}) \cdot Q(\eps)\cdot Q(-\eps^{-1}),\\
B &\coloneqq \biggl(\begin{matrix} \eps^{2} & 1 \\ -1 & 0 \end{matrix}\biggr) = Q(1) \cdot Q(-1)\cdot Q(1) \cdot Q(\eps^{2}),\\
C &\coloneqq \biggl(\begin{matrix} \eps^{-1} & 0 \\ 0 & \eps \end{matrix}\biggr) = Q(-\eps^{-1}) \cdot Q(\eps - 1)\cdot Q(1) \cdot Q(\eps^{-1} -1).
\intertext{Then one can check that
\[
H \ \coloneqq \ A\cdot F \cdot B \cdot F \cdot C \ \in \ A \cdot (Q(f) + \Oh(\eps^3)) \cdot B \cdot (Q(f) + \Oh(\eps^3)) \cdot C \ \in \ Q(-f^2) + \Oh(\eps).
\]
To obtain $H' \in Q(f^2) + \Oh(\eps)$, we replace $B$ by}
B' &\coloneqq \biggl(\begin{matrix} -\eps^{2} & 1 \\ -1 & 0 \end{matrix}\biggr) = Q(1) \cdot Q(-1)\cdot Q(1) \cdot Q(-\eps^{2}).
\end{align*}
One checks that the highest power of $\eps$ appearing in $H$ and $H'$ is at most $2\cdot e_f + 4$. See Fig.\nobreakspace \ref {fig:squaring} and Fig.\nobreakspace \ref {fig:squaringsub} for a pictorial description.
\end{proof}

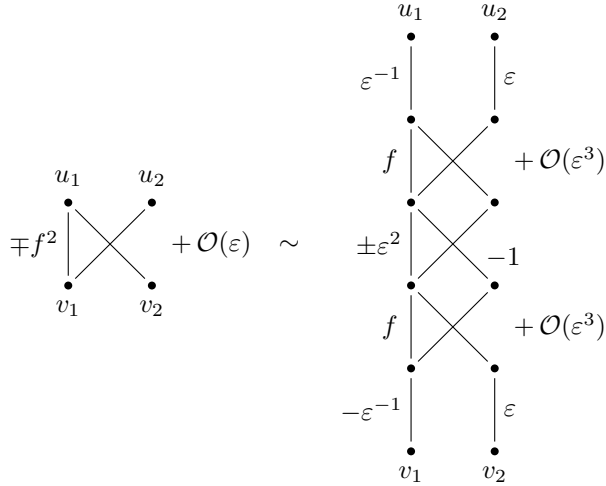
\begin{figure}[h!]
\centering
\begin{minipage}{7em}
\begin {tikzpicture}
\node[state] (11) [label={$u_1$}]{};
\node[state] (12) [right=of 11, label={$u_2$}] {};
\node[state] (21) [below=of 11, label={-90:$v_1$}] {};
\node[state] (22) [below=of 12, label={-90:$v_2$}] {};
\path (11) edge (22);
\path (12) edge (21);
\path (11) edge node[left, xshift=-1pt] {$\mp f^2$} (21);
\path (12) edge[draw=none] node[right] {$\hspace{0.5em}+\, \Oh(\eps)$} (22);
\end{tikzpicture}
\end{minipage}
\hspace{3em}$\sim$\hspace{1em}
\begin{minipage}{10em}
\begin {tikzpicture}
\node[state] (01) [label={$u_1$}]{};
\node[state] (02) [right=of 01, label={$u_2$}] {};
\node[state] (11) [below=of 01]{};
\node[state] (12) [right=of 11] {};
\node[state] (21) [below=of 11] {};
\node[state] (22) [below=of 12] {};
\node[state] (31) [below=of 21] {};
\node[state] (32) [below=of 22] {};
\node[state] (41) [below=of 31] {};
\node[state] (42) [below=of 32] {};
\node[state] (51) [below=of 41, label={-90:$v_1$}] {};
\node[state] (52) [below=of 42, label={-90:$v_2$}] {};
\path (01) edge node[left, xshift=-1pt] {$\eps^{-1}$} (11);
\path (02) edge node[right, xshift=1pt] {$\eps$} (12);
\path (11) edge (22);
\path (12) edge (21);
\path (12) edge[draw=none] node[right] {$\hspace{0.5em}+\, \Oh(\eps^3)$} (22);
\path (11) edge node[left, xshift=-1pt] {$f$} (21);
\path (21) edge node[right, xshift=10pt, yshift=-5pt] {$-1$} (32);
\path (22) edge (31);
\path (21) edge node[left, xshift=-1pt] {$\pm\eps^2$} (31);
\path (31) edge (42);
\path (32) edge (41);
\path (32) edge[draw=none] node[right] {$\hspace{0.5em}+\, \Oh(\eps^3)$} (42);
\path (31) edge node[left, xshift=-1pt] {$f$} (41);
\path (41) edge node[left, xshift=-1pt] {$-\eps^{-1}$} (51);
\path (42) edge node[right, xshift=1pt] {$\eps$} (52);
\end{tikzpicture}
\end{minipage}
\caption{Squaring construction for \protect \MakeUppercase {L}emma\nobreakspace \ref {lem:squaring}}\label{fig:squaring}
\end{figure}
\begin{figure}[h!]
\centering
\begin{minipage}{10em}
\begin {tikzpicture}
\node[state] (01) [label={$u_1$}]{};
\node[state] (02) [right=of 01, label={$u_2$}] {};
\node[state] (11) [below=of 01, label={-90:$v_1$}]{};
\node[state] (12) [right=of 11, label={-90:$v_2$}] {};
\node[state] (21) [below=of 11, yshift=-3em, label={$u_1$}] {};
\node[state] (22) [right=of 21, label={$u_2$}] {};
\node[state] (31) [below=of 21] {};
\node[state] (32) [below=of 22] {};
\node[state] (41) [below=of 31] {};
\node[state] (42) [below=of 32] {};
\node[state] (51) [below=of 41] {};
\node[state] (52) [below=of 42] {};
\node[state] (61) [below=of 51, label={-90:$v_1$}] {};
\node[state] (62) [below=of 52, label={-90:$v_2$}] {};
\path (01) edge node[left, xshift=-1pt] {$\eps^{-1}$} (11);
\path (02) edge node[right, xshift=1pt] {$\eps$} (12);
\path (12) edge[draw=none] node[rotate=90] {$\sim$} (21);
\path (21) edge (32);
\path (22) edge (31);
\path (21) edge node[left, xshift=-1pt] {$\eps^{-1} - 1$} (31);
\path (31) edge (42);
\path (32) edge (41);
\path (31) edge node[left, xshift=-1pt] {$1$} (41);
\path (41) edge (52);
\path (42) edge (51);
\path (41) edge node[left, xshift=-1pt] {$\eps-1$} (51);
\path (51) edge node[left, xshift=-1pt] {$-\eps^{-1}$} (61);
\path (51) edge (62);
\path (52) edge (61);
\end{tikzpicture}
\end{minipage}
\hspace{1em}
\begin{minipage}{10em}
\begin {tikzpicture}
\node[state] (01) [label={$u_1$}]{};
\node[state] (02) [right=of 01, label={$u_2$}] {};
\node[state] (11) [below=of 01, label={-90:$v_1$}]{};
\node[state] (12) [right=of 11, label={-90:$v_2$}] {};
\node[state] (21) [below=of 11, yshift=-3em, label={$u_1$}] {};
\node[state] (22) [right=of 21, label={$u_2$}] {};
\node[state] (31) [below=of 21] {};
\node[state] (32) [below=of 22] {};
\node[state] (41) [below=of 31] {};
\node[state] (42) [below=of 32] {};
\node[state] (51) [below=of 41] {};
\node[state] (52) [below=of 42] {};
\node[state] (61) [below=of 51, label={-90:$v_1$}] {};
\node[state] (62) [below=of 52, label={-90:$v_2$}] {};
\path (01) edge node[left, xshift=-1pt] {$\pm\eps^{2}$} (11);
\path (01) edge node[right, xshift=10pt, yshift=-5pt] {$-1$} (12);
\path (02) edge (11);
\path (12) edge[draw=none] node[rotate=90] {$\sim$} (21);
\path (21) edge (32);
\path (22) edge (31);
\path (21) edge node[left, xshift=-1pt] {$\pm\eps^2$} (31);
\path (31) edge (42);
\path (32) edge (41);
\path (31) edge node[left, xshift=-1pt] {$1$} (41);
\path (41) edge (52);
\path (42) edge (51);
\path (41) edge node[left, xshift=-1pt] {$-1$} (51);
\path (51) edge node[left, xshift=-1pt] {$1$} (61);
\path (51) edge (62);
\path (52) edge (61);
\end{tikzpicture}
\end{minipage}
\begin{minipage}{10em}
\begin {tikzpicture}
\node[state] (01) [label={$u_1$}]{};
\node[state] (02) [right=of 01, label={$u_2$}] {};
\node[state] (11) [below=of 01, label={-90:$v_1$}]{};
\node[state] (12) [right=of 11, label={-90:$v_2$}] {};
\node[state] (21) [below=of 11, yshift=-3em, label={$u_1$}] {};
\node[state] (22) [right=of 21, label={$u_2$}] {};
\node[state] (31) [below=of 21] {};
\node[state] (32) [below=of 22] {};
\node[state] (41) [below=of 31] {};
\node[state] (42) [below=of 32] {};
\node[state] (51) [below=of 41, label={-90:$v_1$}] {};
\node[state] (52) [below=of 42, label={-90:$v_2$}] {};
\path (01) edge node[left, xshift=-1pt] {$-\eps^{-1}$} (11);
\path (02) edge node[right, xshift=1pt] {$\eps$} (12);
\path (12) edge[draw=none] node[rotate=90] {$\sim$} (21);
\path (21) edge (32);
\path (22) edge (31);
\path (21) edge node[left, xshift=-1pt] {$-\eps^{-1}$} (31);
\path (31) edge (42);
\path (32) edge (41);
\path (31) edge node[left, xshift=-1pt] {$\eps$} (41);
\path (41) edge (52);
\path (42) edge (51);
\path (41) edge node[left, xshift=-1pt] {$-\eps^{-1}$} (51);
\end{tikzpicture}
\end{minipage}
\caption{Squaring construction subroutines for $C$, $B$, and $A$ for \protect \MakeUppercase {L}emma\nobreakspace \ref {lem:squaring}}\label{fig:squaringsub}
\end{figure}
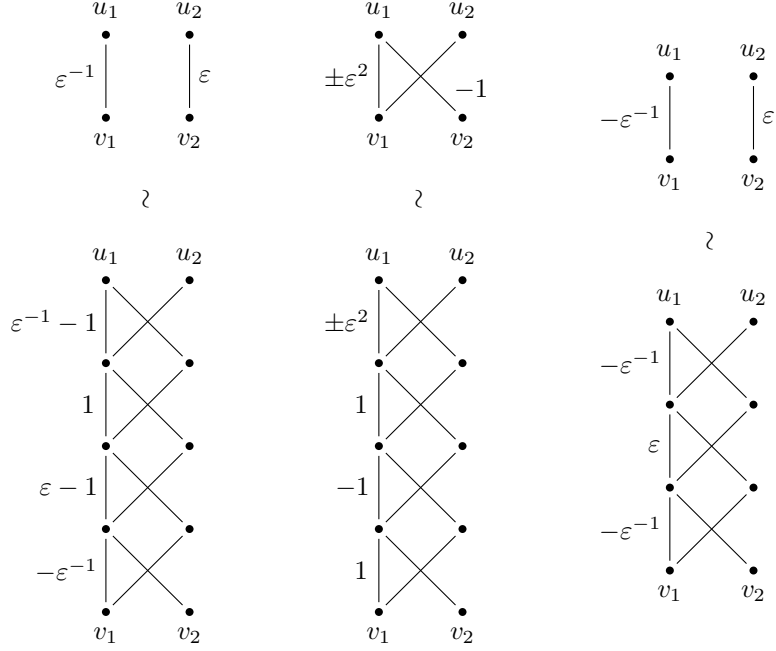

\begin{lemma}[Multiplication]\label{lem:mult}
Let $f,g \in \FF[\xvars]$ be polynomials such that some $F \in Q(f/2) + \Oh(\eps^3)$ and some $G \in Q(g) + \Oh(\eps^3)$ can be written as the product of $n$ and $m$ primitive Q-matrices respectively. Then some $H \in Q(f\cdot g) + \Oh(\eps)$ can be written as the product of $4n + 4m + 37$ primitive Q-matrices. Moreover, if the error degrees in $F$, $G$ are $e_f$, $e_g$, respectively, then the error degree of $H$ is at most $4\cdot e_f + 4\cdot e_g + 12$.
\end{lemma}
\begin{proof}
We make use of the identity $( -(f/2)^2) + (-g^2) + (f/2 + g)^2 = f\cdot g$. By the addition lemma (\protect \MakeUppercase {L}emma\nobreakspace \ref {lem:addition}), $(f/2 + g) + \Oh(\eps^3)$ can be written as the product of $n + m + 1$ primitive Q-matrices with error degree at most $e_f + e_g$. By the squaring lemma (\protect \MakeUppercase {L}emma\nobreakspace \ref {lem:squaring}), $Q(-(f/2)^2) + \Oh(\eps)$, $Q(-g^2) + \Oh(\eps)$, and $Q((f/2 + g)^2) + \Oh(\eps)$ can be written as the product of $2n + 11$, $2m + 11$, and $2(n + m + 1) + 11$ primitive Q-matrices, respectively. The corresponding error degrees are at most $2\cdot e_f + 4$, $2\cdot e_g + 4$, and $2(e_f + e_g) + 4$. Finally, by the addition lemma again, $Q(f\cdot g)+\Oh(\eps) = Q( -(f/2)^2 + (-g^2) + (f/2 + g)^2) + \Oh(\eps)$ can be written as the product of  $(2n + 11) + 1 + (2m + 11) + 1 + (2(n + m + 1) + 11) = 4n + 4m + 37$ primitive Q-matrices. The corresponding error degree is at most $(2\cdot e_f + 4) + (2\cdot e_g + 4) + (2 (e_f + e_g) + 4) = 4\cdot e_f + 4 \cdot e_g + 12$.
 See Fig.\nobreakspace \ref {fig:mult} for a pictorial description.
\end{proof}

\begin{figure}[h!]
\centering
\begin{minipage}{7em}
\begin {tikzpicture}
\node[state] (11) [label={$u_1$}]{};
\node[state] (12) [right=of 11, label={$u_2$}] {};
\node[state] (21) [below=of 11, label={-90:$v_1$}] {};
\node[state] (22) [below=of 12, label={-90:$v_2$}] {};
\path (11) edge (22);
\path (12) edge (21);
\path (11) edge node[left, xshift=-1pt] {$f \cdot  g$} (21);
\path (12) edge[draw=none] node[right] {$\hspace{0.5em}+\, \Oh(\eps)$} (22);
\end{tikzpicture}
\end{minipage}
\hspace{4em}$\sim$\hspace{1em}
\begin{minipage}{10em}
\begin {tikzpicture}
\node[state] (11) [label={$u_1$}]{};
\node[state] (12) [right=of 11, label={$u_2$}] {};
\node[state] (21) [below=of 11] {};
\node[state] (22) [below=of 12] {};
\node[state] (31) [below=of 21] {};
\node[state] (32) [below=of 22] {};
\node[state] (41) [below=of 31] {};
\node[state] (42) [below=of 32] {};
\node[state] (51) [below=of 41] {};
\node[state] (52) [below=of 42] {};
\node[state] (61) [below=of 51, label={-90:$v_1$}] {};
\node[state] (62) [below=of 52, label={-90:$v_2$}] {};
\path (11) edge (22);
\path (12) edge (21);
\path (12) edge[draw=none] node[right] {$\hspace{0.5em}+\, \Oh(\eps)$} (22);
\path (11) edge node[left, xshift=-1pt] {$-(f/2)^2$} (21);
\path (21) edge (32);
\path (22) edge (31);
\path (31) edge (42);
\path (32) edge[draw=none] node[right] {$\hspace{0.5em}+\, \Oh(\eps)$} (42);
\path (31) edge node[left, xshift=-1pt] {$-g^2$} (41);
\path (32) edge (41);
\path (41) edge (52);
\path (42) edge (51);
\path (51) edge (62);
\path (52) edge (61);
\path (52) edge[draw=none] node[right] {$\hspace{0.5em}+\, \Oh(\eps)$} (62);
\path (51) edge node[left, xshift=-1pt] {$(f/2 + g)^2$} (61);
\end{tikzpicture}
\end{minipage}
\caption{Multiplication construction for \protect \MakeUppercase {L}emma\nobreakspace \ref {lem:mult}}\label{fig:mult}
\end{figure}
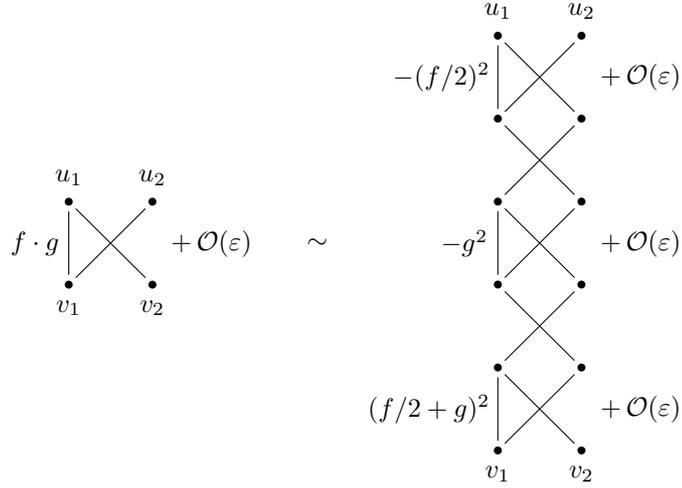

\begin{proposition}\label{prop:constructQ}
Let $f$ be a polynomial computed by a formula of depth $d$. For every constant $\alpha \in \FF$, some matrix in $F \in Q(\alpha f) + \Oh(\eps)$ can be written as a product of at most $45\cdot 9^d$ primitive Q-matrices. Moreover, $F$ has error degree at most $12\cdot 25^d$.
\end{proposition}
\begin{proof}
The proof is by induction on $d$. For $d = 0$, that is, $f$ is a constant $\beta \in \FF$ or a variable~$x$, note that $Q(f)$ can be written directly as a primitive Q-matrix (with error degree~0). Since also $Q(\alpha/2)$ can be written directly (also with error degree 0), we can use the multiplication lemma (\protect \MakeUppercase {L}emma\nobreakspace \ref {lem:mult}), to write $Q(\alpha f) + \Oh(\eps)$ as a product of $4 + 4 + 37 = 45$ primitive Q-matrices (with error degree at most 12).

For $d\geq 1$, fix a constant $\alpha$. We know that either $f = g + h$ or $f = g\cdot h$ with formulas $g,h$ of depth $< d$. By the induction hypothesis, for any constant $\beta, \gamma$, we can write $Q(\beta g) + \Oh(\eps)$ and $Q(\gamma h) + \Oh(\eps)$ as a product of $n_g, n_h \leq 45\cdot 9^{d-1}$ primitive Q-matrices, with error degrees $e_g, e_h \leq 12\cdot 25^{d-1}$.

\emph{Case $f = g+ h$.} We set $\beta = \gamma = \alpha$ and use the addition lemma (\protect \MakeUppercase {L}emma\nobreakspace \ref {lem:addition}) to obtain $Q(\alpha f) + \Oh(\eps) = Q(\alpha g + \alpha h) + \Oh(\eps)$ as a product of $n_g + n_h + 1 \leq 2 \cdot 45 \cdot 9^{d-1} + 1 \leq 45 \cdot 9^d$ primitive Q-matrices, with error degree at most $e_g + e_h \leq 2\cdot 12 \cdot 25^{d-1} \leq 12 \cdot 25^d$.

\emph{Case $f = g\cdot h$.} By replacing $\eps$ by $\eps^3$ in all primitive Q-matrices, we obtain matrices in $Q(\beta g) + \Oh(\eps^3)$ and $Q(\gamma h) + \Oh(\eps^3)$ as a product of $n_g$ and $n_h$ primitive Q-matrices with error degree at most $3\cdot e_g$ and $3\cdot e_h$ respectively.
Now we set $\beta = \alpha/2$ and $\gamma = 1$ and use the multiplication lemma (\protect \MakeUppercase {L}emma\nobreakspace \ref {lem:mult}) to obtain $Q(\alpha f) + \Oh(\eps) = Q((\alpha\cdot g)\cdot h) + \Oh(\eps)$ as a product of $4 n_g + 4 n_h + 37 \leq 8 \cdot 45\cdot 9^{d-1} + 37 \leq 45 \cdot 9^d$ primitive Q-matrices. The error degree is at most $4 ( 3\cdot e_g) + 4(3\cdot e_h) + 12 = 12(e_g + e_h + 1) \leq 24 \cdot 12 \cdot 25^{d-1} + 12 \leq 12 \cdot 25^d$.
\end{proof}

\begin{proposition}\label{pro:brent}
If $(f_n) \in \VPe$, then for each $n$ 
a matrix in $F \in Q(f_n) + \Oh(\eps)$ can be written as a product of $\poly(n)$ many primitive Q-matrices. Moreover, $F$ has error degree at most $\poly(n)$.
\end{proposition}
\begin{proof}
The construction uses the classical depth-reduction theorem for formulas by Brent~\cite{brent1974parallel}, for which a modern proof can be found in the survey of Saptharishi \cite[Lemma~5.5]{ramprasad}:
If a family $(f_n)$ has polynomially bounded formula size, then
there are formulas computing~$f_n$ that have size $\poly(n)$ and depth $\Oh(\log n)$.
Applying Proposition\nobreakspace \ref {prop:constructQ} now yields the result.
\end{proof}

\begin{proof}[\upshape\bfseries Proof of \protect \MakeUppercase {T}heorem\nobreakspace \ref {approxvpe}]
This follows directly from Proposition\nobreakspace \ref {pro:brent}. Namely, let $(f_n) \in \VP_e$.
By Proposition\nobreakspace \ref {pro:brent} there is an $F \in Q(f_n)+ \Oh(\eps)$ which is a product of polynomially many primitive Q-matrices such that $F$ has polynomially bounded error degree.
The width-2 ABP computing $f_n+\Oh(\eps)$ is given by
$\begin{psmallmatrix}1 & 0\end{psmallmatrix} F \begin{psmallmatrix}1 \\ 0\end{psmallmatrix}$.
\end{proof}

\begin{example}
Following the construction in \protect \MakeUppercase {T}heorem\nobreakspace \ref {approxvpe} we get the following ABP for approximating the polynomial $x_1 x_2 + x_3 x_4 + \cdots + x_{15} x_{16}$, which cannot be computed by any width-2 ABP.
Let 
\begin{multline*}
F(x_1, x_2) = \begin{pmatrix} \frac{1}{\eps}-\frac{\eps x_1}{2} & -\frac{x_1}{2 \eps} \\
 \eps^3 & \eps \\
\end{pmatrix}
\begin{pmatrix} \frac{1}{2} (x_1-2 x_2) \eps^2+1 & \frac{1}{2} (x_1-2 x_2) \\
 \eps^2 & 1 \\
\end{pmatrix}\\
\cdot
\begin{pmatrix} \frac{x_1 \eps^2}{2}+1 & -\frac{x_1}{2} \\
 -\eps^2 & 1 \\
\end{pmatrix}\begin{pmatrix} \frac{x_1+2 x_2}{2 \eps} & \eps \\
 \eps^{-1} & 0 \\
\end{pmatrix}.
\end{multline*}
Then 
\[
F(x_1,x_2) = \begin{pmatrix} x_1 x_2 & 1 \\ 1 & 0 \end{pmatrix} + \Oh(\eps).
\]
Using the addition lemma \protect \MakeUppercase {L}emma\nobreakspace \ref {lem:addition} we get 
\[
\begin{psmallmatrix}1 & 0\end{psmallmatrix} F(x_1,x_2)  \begin{psmallmatrix}0 & 1\\ 1 & 0\end{psmallmatrix} F(x_3,x_4) \cdots  \begin{psmallmatrix}0 & 1\\ 1 & 0\end{psmallmatrix} F(x_{15},x_{16}) \begin{psmallmatrix}1 \\ 0\end{psmallmatrix} = x_1 x_2 + x_3x_4 + \cdots + x_{15} x_{16} + \Oh(\eps).
\]
\end{example}

\begin{corollary}\label{approxcor}
$\approxbar{\VPk{2}} = \approxbar{\VPe}$ and $\polyapproxbar{\VPk{2}} = \polyapproxbar{\VPe}$ when $\characteristic(\FF)\neq2$.
\end{corollary}
\begin{proof}
The inclusion $\VPk{2} \subseteq \VPe$ is standard (see Proposition\nobreakspace \ref {prop:vpkinvpe}).
Taking closures on both sides, we obtain $\approxbar{\VPk{2}} \subseteq \approxbar{\VPe}$ and $\polyapproxbar{\VPk{2}} \subseteq \polyapproxbar{\VPe}$.

On the other hand, when $\characteristic(\FF) \neq 2$, we have the inclusion $\VPe \subseteq \polyapproxbar{\VPk{2}}$ (\protect \MakeUppercase {T}heorem\nobreakspace \ref {approxvpe}).
By taking closures this implies $\approxbar{\VPe} \subseteq \approxbar{\VPk{2}}$ and $\polyapproxbar{\VPe} \subseteq \polyapproxbar{\VPk{2}}$.
\end{proof}

\begin{corollary}\label{cor:interpolation}
$\polyapproxbar{\VPk{2}} = \VPe$ when $\characteristic(\FF) \neq 2$ and $\FF$ is infinite.
\end{corollary}
\begin{proof}
By \protect \MakeUppercase {C}orollary\nobreakspace \ref {approxcor} we have $\polyapproxbar{\VPk{2}} = \polyapproxbar{\VPe}$.
It remains to show the equality $\polyapproxbar{\VPe} = \VPe$. We give a proof of this via a standard interpolation argument in Section\nobreakspace \ref {sec:cor:interpolation}.
\end{proof}

As a consequence of Proposition\nobreakspace \ref {prop:constructQ}, we obtain a new description of $\approxbar{\VPe}$ as follows. We define the \defin{generalized Fibonacci polynomial} $F_n(x_1,\ldots, x_n)$ by
$F_0 \coloneqq 1$, $F_1 \coloneqq x_1$, and $F_n \coloneqq x_n F_{n-1} + F_{n-2}$ for all $n\geq 2$.
The name comes from the fact that $F_n(1, 1, \ldots, 1)$ is the $n$th Fibonacci number and $F_n(x, x, \ldots, x)$ is the $n$th Fibonacci polynomial. Another description of the polynomial $F_n$ is that it is the upper left entry of a product of Q-matrices~$Q(x_i)$, that is, $F_n(x_1, \ldots, x_n) = (Q(x_n) Q(x_{n-1}) \cdots Q(x_1))_{1,1}$.

\begin{definition}\label{def:fibcompl}
A polynomial $f$ is a \defin{projection} of $F_m$ if there exist affine linear forms $\ell_1, \ldots, \ell_m$ such that $f = F_m(\ell_1, \ldots, \ell_m)$. The smallest $m$ such that $f$ is a projection of $F_m$ we call the \defin{Fibonacci complexity of $f$}. A polynomial is a \defin{degeneration of $F_m$} if there exist parametrized affine linear forms $\ell_1(\eps), \ldots, \ell_m(\eps)$ such that $f = F_m(\ell_1(\eps), \ldots, \ell_m(\eps))$. The smallest~$m$ such that $f+ \Oh(\eps)$ is a degeneration of $F_m$ we call the \defin{border Fibonacci complexity} of $f$, and is denoted by $\approxcompl_\Fib(f)$.
A family $(h_n)$ of polynomials is called \emph{$\approxbar{\VPe}$-complete under p-degenerations} if $(h_n) \in \approxbar{\VPe}$ and for every $(f_n) \in \approxbar{\VPe}$ there exists a polynomially bounded function $t$ such that some polynomial in $f_n + \Oh(\eps)$ is a degeneration of $F_{t(n)}$.
\end{definition}

The Fibonacci complexity is not always finite (\cite{wang}), but Proposition\nobreakspace \ref {pro:brent} shows that the border Fibonacci complexity $\approxcompl_\Fib(f)$ is always finite
and that $\approxbar{\VPe}$ can be characterized as the class of families with polynomially bounded border Fibonacci complexity:
\begin{theorem}\label{thm:fibon}
$\approxbar{\VPe} = \{ (f_n) \mid \approxcompl_\Fib(f_n) \in \polyf(n)\}$.
\end{theorem}
\begin{proof}
Clearly the right-hand side is contained in the left-hand side.
$\VPe$ is contained in the right-hand side by Proposition\nobreakspace \ref {pro:brent}.
A moment's thought reveals that the right-hand side is closed under the approximation closure in the sense of \protect \MakeUppercase {D}efinition\nobreakspace \ref {barclass}.
Thus taking the closure on both sides yields the result.
\end{proof}

\protect \MakeUppercase {T}heorem\nobreakspace \ref {thm:fibon} says that $(F_n)$ is $\approxbar{\VPe}$-complete under p-degenerations. From the proof of Proposition\nobreakspace \ref {prop:constructQ} it follows that also $(F_{2n+1})$ is $\approxbar{\VPe}$-complete under p-degenerations, that is, we only need the $F_m$ with odd index~$m$ (this follows from $\det(Q(f))=-1)$.

\begin{remark}[Symmetry]
Define the polynomial 
$C_n(x_1, \ldots, x_n)$ as
\[
C_n(x_1, \ldots, x_n) \coloneqq \mathrm{trace}( Q(x_n) \cdot Q(x_{n-1}) \cdots Q(x_1)).
\]
Since the trace of a matrix product is invariant under cyclic shifts of the matrices, the polynomial $C_n(x_1, \ldots, x_n)$ is invariant under cyclic shifts of the variables $x_1, \ldots, x_n$. Thus~$C_n$ can be viewed as a cyclically symmetric version of $F_n$. (Note that $C_n$ and $F_n$ are also both invariant under reversing the order of the variables $x_1, \ldots, x_n$, that is, mapping $(x_1, \ldots, x_n)$ to $(x_n, \ldots, x_1)$.)

Define the \defin{border cyclic Fibonacci complexity} analogously to the border Fibonacci complexity by replacing $F_n$ by $C_n$ in \protect \MakeUppercase {D}efinition\nobreakspace \ref {def:fibcompl}.
Analogously to \protect \MakeUppercase {T}heorem\nobreakspace \ref {thm:fibon} we now see that the families $(C_n)$ and $(C_{2n+1})$ are both $\approxbar{\VPe}$-complete under p-degenerations.
\end{remark}

\begin{remark}[A closed form for $F_n$ and $C_n$]
We describe another way to write $F_n$ and~$C_n$.
An \emph{adjacent pair} is a set of two numbers $\{i,i+1\}$ with $1 \leq i < n$.
A \emph{supporting set} is the set $\{1,2,\ldots,n\}$ after removing a disjoint (possibly empty) union of adjacent pairs.
For a supporting set~$S$ define $x_S \coloneqq \prod_{i \in S} x_i$.
Then $F_n(x_1, \ldots, x_n) = \sum_S x_S$, where the sum is over all supporting sets.

We define a \emph{cyclicly adjacent pair} as a set that is either an adjacent pair or the set $\{1,n\}$, if $1 \neq n$.
We define a \emph{cyclic supporting set} as the set $\{1,2,\ldots,n\}$ after removing a disjoint (possibly empty) union of cyclicly adjacent pairs.
Then
$C_n(x_1, \ldots, x_n) = \sum_S x_S$, where the sum is over all cyclic supporting sets.
\end{remark}

\begin{remark}[Planarity]
We remark that the product of two Q-matrices $Q(x)Q(y)$ can be rewritten as $Q(x)Q(y) = \bigl(Q(x)\begin{psmallmatrix}0&1\\1&0\end{psmallmatrix}\bigr)\bigl(\begin{psmallmatrix}0&1\\1&0\end{psmallmatrix}Q(y)\bigr)$. We also have $Q(x)\begin{psmallmatrix}a\\b\end{psmallmatrix} = \bigl(Q(x)\begin{psmallmatrix}0&1\\1&0\end{psmallmatrix}\bigr) \begin{psmallmatrix}b\\a\end{psmallmatrix}$.
Consider a width-2 ABP that is a product of primitive Q-matrices,
\[
\begin{psmallmatrix}a & b\end{psmallmatrix} Q(\ell_1) Q(\ell_2) \cdots Q(\ell_k) \begin{psmallmatrix}c \\ d\end{psmallmatrix}.
\]
By pairing up the $i$th Q-matrix with the $(i+1)$th Q-matrix for each odd $i$, and using the above equations, we can rewrite this ABP into a width-2 ABP whose underlying graph has no crossing edges, that is, a \emph{planar} with-2 ABP. See Fig.\nobreakspace \ref {fig:planar} for an example with three Q-matrices.
\end{remark}

\begin{figure}[h!]
\centering
\begin{minipage}{7em}
\begin{tikzpicture}[node distance=1.1cm and 0.6*1.1cm]
\node[state] (21)  {};
\node (22) [right=of 21] {};
\node[state] (31) [below left=of 21] {};
\node[state] (32) [below=of 22] {};
\node[state] (41) [below=of 31] {};
\node[state] (42) [below=of 32] {};
\node[state] (51) [below=of 41] {};
\node[state] (52) [below=of 42] {};
\node[state] (61) [below=of 51] {};
\node[state] (62) [below=of 52] {};
\node[state] (71) [below right=of 61] {};
\path (21) edge node[left, xshift=-0.4em] {$a$} (31);
\path (21) edge node[right, xshift=0.4em] {$b$} (32);
\path (32) edge (41);
\path (31) edge (42);
\path (41) edge (52);
\path (42) edge (51);
\path (31) edge node[left] {$\ell_1$}(41);
\path (41) edge node[left] {$\ell_2$}(51);
\path (51) edge node[left] {$\ell_3$} (61);
\path (51) edge (62);
\path (52) edge (61);
\path (62) edge node[right, xshift=0.4em] {$d$} (71);
\path (61) edge node[left, xshift=-0.4em] {$c$} (71);
\end{tikzpicture}
\end{minipage}
\hspace{1em}$\sim$\hspace{1em}
\begin{minipage}{7em}
\begin{tikzpicture}[node distance=1.1cm and 0.6*1.1cm]
\node[state] (21)  {};
\node (22) [right=of 21] {};
\node[state] (31) [below left=of 21] {};
\node[state] (32) [below=of 22] {};
\node[state] (41) [below=of 31] {};
\node[state] (42) [below=of 32] {};
\node[state] (51) [below=of 41] {};
\node[state] (52) [below=of 42] {};
\node[state] (61) [below=of 51] {};
\node[state] (62) [below=of 52] {};
\node[state] (71) [below right=of 61] {};
\path (21) edge node[left, xshift=-0.4em] {$a$} (31);
\path (21) edge node[right, xshift=0.4em] {$b$} (32);
\path (32) edge (42);
\path (31) edge (41);
\path (41) edge (51);
\path (42) edge (52);
\path (31) edge node {$\ell_1$}(42);
\path (42) edge node {$\ell_2$}(51);
\path (51) edge node {$\ell_3$}(62);
\path (51) edge (61);
\path (52) edge (62);
\path (62) edge node[right, xshift=0.4em] {$c$} (71);
\path (61) edge node[left, xshift=-0.4em] {$d$} (71);
\end{tikzpicture}
\end{minipage}
\caption{Making an ABP consisting of three primitive Q-matrices planar.}\label{fig:planar}
\end{figure}
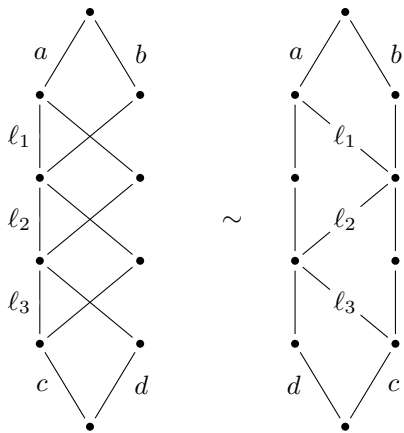

\section{VNP via products of affine linear forms}\label{sec:secondmain}
Valiant proved the following characterization of $\VNP$ \cite{valiant1980reducibility} (see also \cite[Thm.~21.26]{burgisser1997algebraic},\\ \cite[Thm.~2.13]{burgisser2000completeness} and \cite[Thm.~2]{malod2008}).

\begin{theorem}[Valiant \cite{valiant1980reducibility}]\label{valiant}
$\VNPe = \VNP$.
\end{theorem}

We strengthen Valiant's characterization of $\VNP$ from $\VNPe$  to $\VNPk{1}$.

\begin{theorem}\label{VNPmain}
$\VNP_{1} = \VNP$ when $\characteristic(\FF) \neq 2$.
\end{theorem}

We give two proofs.
The idea of the first proof is to show that the $\VNP$-complete permanent family
$\per_n \coloneqq \sum_{\sigma \in S_n} \prod_{i\in [n]} x_{i,\sigma(i)}$
is in $\VNP_1$.
The idea of the second proof is to simulate in $\VNP_1$ the primitives that are used in the proof of $\VPe = \VPk{3}$ by \cite{ben1992computing}.
We present the second proof in Section\nobreakspace \ref {sec:VNPdirect}.
The advantage of the second proof is that we can restrict the ABP edge labels to affine linear forms that have at most 2 variables, see \protect \MakeUppercase {T}heorem\nobreakspace \ref {thm:secondmainstronger}.
Both proofs use the following lemma to write expressions of the form $1+xy$ as a hypercube sum of a product of affine linear forms.

\begin{lemma}\label{tricklem}
$\tfrac12 \sum_{b\in \{0,1\}} (x + 1 - 2b) (y + 1 - 2b) = 1+xy$ when $\characteristic(\FF)\neq 2$.
\end{lemma}
\begin{proof}
Expanding the left side gives the right side.
\end{proof}

\begin{proof}[\bfseries\upshape Proof of \protect \MakeUppercase {T}heorem\nobreakspace \ref {VNPmain}]
The permanent family $(\per_n)$ is well-known to be $\VNP$-complete under p-projections, see for example \cite[Thm.~2.10]{burgisser2000completeness}.
Therefore, to show that $\VNP\subseteq \VNPk{1}$, it suffices to show that $(\per_n) \in \VNPk{1}$.
We begin by writing $\per_n$ as an inclusion-exclusion-type expression due to Ryser \cite[Thm.~4.1]{MR0150048},
\[
\per_n = (-1)^n \sum_{S\subseteq [n]} (-1)^{|S|} \prod_{j\in [n]} \sum_{i\in S} x_{i,j}.
\]
Encoding every subset $S\subseteq [n]$ by a bit string $b = (b[1],\ldots,b[n]) \in \{0,1\}^n$, we can rewrite the above as
\begin{align*}
\per_n &= (-1)^n \sum_{b\in \{0,1\}^n} \Bigl(\prod_{k\in [n]}(1-2b[k])\Bigr) \prod_{j\in [n]} \sum_{i\in [n]}  b[i]\, x_{i,j}\\
&= (-1)^n \sum_{b\in \{0,1\}^n} \Bigl(\prod_{k\in [n]}(1-2b[k])\Bigr) \sum_{i_1, \ldots, i_n\in [n]} \prod_{j\in [n]} b[i_j]\, x_{i_j,j}
\end{align*}
For notational convenience we use square brackets not only to refer to sets ($[n]:=\{1,\ldots,n\}$), but also to entries in a list ($b[k]:=b_k$).
We now introduce new Boolean variables $a[i,j]$, $1\leq i\leq n-1$, $1 \leq j\leq n$, and fix the values $a[0,j] = 1$, $a[n,j] = 0$.
(This gives an $(n+1)\times n$ matrix of variables and constants in which the first row consists of all 1s and the last row contains only 0s.)
We claim that the above expression equals
\begin{align}\label{eqstar}
\per_n &=  (-1)^n  \sum_{b\in \{0,1\}^n}\bigg(\prod_{k\in [n]}(1-2b[k]) \cdot \sum_a \prod_{i,j\in [n]} \bigl(1 + (x_{i,j}-1)(a[i\compactminus1,j] - a[i,j])\bigr)\\
&  \cdot \bigl(1 + (b[i]-1)(a[i\compactminus1,j] - a[i,j])\bigr)\cdot \bigl(1 + (a[i\compactminus1,j] - 1)a[i,j]\bigr) \bigg),\notag
\end{align}
where the second sum is over all Boolean assignments of $a[i,j]$.
The idea is to encode the indices $i_1,\ldots,i_n$ in the boolean variables $a[i,j]$ in unary. For example, for $n=4$, if $i_1=4$, $i_2=3$, $i_3=1$, $i_4=4$,
then the corresponding matrix $a$ is
\[
\begin{pmatrix}
1&1&1&1\\
1&1&0&1\\
1&1&0&1\\
1&0&0&1\\
0&0&0&0
\end{pmatrix}.
\]
We prove the claim \eqref{eqstar} in three steps. Fix $j$.
\begin{itemize}
\item  If $a[i-1,j] = 0$ and $a[i,j] = 1$, then $1 + (a[i-1,j]-1)a[i,j] = 0$. Thus if in the sequence $a[0,j], \ldots, a[n,j]$ a 0 is followed by a 1, then $\prod_{i\in [n]} (1 + (a[i-1,j]-1)a[i,j] = 0$.
Conversely, if $(a[0,j], \ldots, a[n,j]) = (1,\ldots, 1, 0,\ldots, 0)$, then $\prod_{i\in [n]} (1 + (a[i-1,j]-1)a[i,j]) = 1$.
The nontrivial assignments of $(a[0,j],\ldots, a[n,j])$ are thus exactly of the form $(1,\ldots, 1,0,\ldots, 0)$ where the first 0 occurs at some index $1\leq z \leq n$ (since we have set $a[0,j] = 1$ and $a[n,j] = 0$).
Fix such an assignment with first 0 occurring at index $z$.
\item If $i=z$, then $1+(x_{i,j}-1)(a[i-1,j]-a[i,j])$ equals $x_{i,j}$. If $i\neq z$, it equals $1$.
\item If $i=z$, then $1+(b[i]-1)(a[i-1,j]-a[i,j])$ equals $b[i]$. If $i\neq z$, it equals $1$.
\end{itemize}
This proves \eqref{eqstar}.

Next we apply \protect \MakeUppercase {L}emma\nobreakspace \ref {tricklem}, introducing fresh hypercube variables $c_1[i,j]$, $c_2[i,j]$, and $c_3[i,j]$, for $1\leq i,j\leq n$, to obtain
\begin{align*}
\per_n &= (-1)^n(\tfrac{1}{2})^{3n^2} \sum_{b} \Bigl(\prod_{k\in [n]}(1-2\,b[k])\Bigr) \cdot \sum_{a} \bigg( \prod_{i,j\in [n]}\\
&{}\phantom{{}\cdot{}}\sum_{c_1[i,j]}\Big[ (x_{i,j} - 2\, c_1[i,j]) 
\cdot (a[i\compactminus1,j] - a[i,j] + 1-2\,c_1[i,j])\Big]\\
&{}\cdot{} \sum_{c_2[i,j]}\Big[ (b[i] - 2\,c_2[i,j])
\cdot (a[i\compactminus1, j] - a[i,j] + 1-2\,c_2[i,j])\Big]\\
&{}\cdot{} \sum_{c_3[i,j]}\Big[(a[i\compactminus1, j] - 2\,c_3[i,j])
\cdot (a[i,j] + 1-2\,c_3[i,j])\Big] \bigg),
\end{align*}
where the sum goes over all Boolean assignments of $b[i]$, $a[i,j]$, $c_1[i,j]$, $c_2[i,j]$, $c_3[i,j]$, for all indices $1\leq i,j \leq n$, except for $a[n,j] \coloneqq 0$, and $a[0,j]\coloneqq 1$.
After a rearrangement we obtain the expression
\begin{align*}
\per_n &=  \sum_{\substack{a,b\\c_1,c_2,c_3}} \bigg( (-1)^n (\tfrac{1}{2})^{3n^2} \Bigl(\prod_{k\in [n]}(1-2\,b[k])\Bigr) \cdot  \prod_{i,j\in [n]}\\
&{}\phantom{{}\cdot{}} (x_{i,j} - 2\, c_1[i,j]) \cdot (a[i\compactminus1,j] - a[i,j] + 1-2\,c_1[i,j])\\[0.5ex]
&{}\cdot{} (b[i] - 2\,c_2[i,j]) \cdot (a[i\compactminus1, j] - a[i,j] + 1-2\,c_2[i,j]) \\
&{}\cdot{} (a[i\compactminus1, j] - 2\,c_3[i,j]) \cdot (a[i,j] + 1-2\,c_3[i,j]) \bigg),
\end{align*}
where the sum goes over all Boolean assignments of $a[i,j]$, $b[i]$, $c_1[i,j]$, $c_2[i,j]$, $c_3[i,j]$ for all indices $1\leq i,j \leq n$, again except for $a[n,j] \coloneqq 0$, and $a[0,j]\coloneqq 1$.
This shows that $(\per_n) \in \VNPk{1}$.
\end{proof}


In Section\nobreakspace \ref {pro:FIIproof} we will prove that the statement of \protect \MakeUppercase {T}heorem\nobreakspace \ref {VNPmain} does not hold over $\FF_2$, that is, $\VNPk{1} \subsetneq \VNP$ when $\FF = \FF_2$.
We leave the situation over other fields of characteristic 2 as an open problem.


\section{ABPs with restricted edge labels}\label{sec:differentabpedgelabels}

So far the edge labels of our ABPs were allowed to be arbitrary affine linear forms.
This section is about ABPs in which the edge labels are restricted to be simple affine linear forms (``weak ABPs''), or variables and constants (``weakest ABPs''). These edge label types were also studied in \cite{wang}. 

\begin{definition}\label{modelsubsets}
A $\weakest$-ABP (weakest ABP) is an ABP with edges labeled by variables or constants.
A $\weak$-ABP (weak ABP) is an ABP with edges labeled by simple affine linear forms $\alpha x_i + \beta$, $\alpha,\beta \in \FF$.
A $\gen$-ABP (general ABP) is an ABP with edges labeled by general affine linear forms $\sum_i \alpha_i x_i + \beta$, $\alpha_i, \beta \in \FF$.
For $*$ equal to $\weakest$, $\weak$ or $\gen$, the class $\VPk{k}^*$ consists of all families
of polynomials over polynomially many variables that are
computed by polynomial-size width-$k$ $*$-ABPs.
In the rest of this paper the star will act as a variable from $\{\weakest,\weak,\gen\}$.
We write $\VPk{k}$ if we mean $\VPk{k}^\gen$.
\end{definition}

From the above definition it follows that $\VPk{k}^\weakest\subseteq \VPk{k}^{\weak} \subseteq \VPk{k}^{\gen}$.

\begin{remark}\label{moreprecise}
One checks that the construction in the proof of \protect \MakeUppercase {T}heorem\nobreakspace \ref {approxvpe} actually proves the inclusion  $\VPe \subseteq \polyapproxbar{\VPk{2}^\weakest}$ when $\characteristic(\FF)\neq 2$. The inclusion $\VPe \subseteq \polyapproxbar{\VPk{2}^\weakest}$ implies the equalities $\approxbar{\VPk{2}^\weakest} = \approxbar{\VPe}$ and $\polyapproxbar{\VPk{2}^\weakest} = \polyapproxbar{\VPe}$.
\end{remark}

In the following sections we will prove all inclusions and separations that are listed in Figure~\ref{fig:overview}.

\subsection{Comparing different types of edge labels in width-2 ABPs}\label{subsec:comp2}

The aim of this subsection is to prove the following separation.

\begin{theorem}\label{th:sepweakgen}
$\VPk{2}^\weak \subsetneq \VPk{2}^\gen$.
\end{theorem}

In fact, we will show the following stronger statement.

\begin{theorem}\label{notwabp}
The polynomial
\begin{align*}
p(\xvars) &= (x_{11} + x_{12} + \cdots + x_{17})(x_{21} + x_{22} + \cdots + x_{27})\\ &\phantom{{}={}}{}+ (x_{31} + x_{32} + \cdots + x_{37})(x_{41} + x_{42} + \cdots + x_{47})
\end{align*}
is computable by a width-2 $\gen$-ABP, but not computable by any width-2 $\weak$-ABP.
\end{theorem}

We leave it as an open problem whether the inclusion $\VPk{2}^\weakest \subseteq \VPk{2}^\weak$ is strict.

To prove \protect \MakeUppercase {T}heorem\nobreakspace \ref {notwabp} we will review and reuse the arguments used by Allender and Wang~\cite{wang} to show that the polynomial $x_1x_2 + \cdots + x_{15}x_{16}$ cannot be computed by any width-2 $\gen$-ABP.

For the proof of \protect \MakeUppercase {T}heorem\nobreakspace \ref {notwabp} we may without loss of generality assume that the base field~$\FF$ is algebraically closed, because for any field $\FF$, if $p$ is not computable over the algebraic closure of $\FF$, then it is not computable over $\FF$ itself.
Let $\HH$ be the affine linear forms that are single variables $x_i$ or constants $\FF$.
Let $\SSS$ be the set of simple affine linear forms. Let $\LL$ be the set of general affine linear forms.
Let $\HH^{2\times 2}$, $\SSS^{2\times 2}$, $\LL^{2\times 2}$ be the sets of $2\times 2$ matrices with entries in $\HH$, $\SSS$, $\LL$ respectively.
In this subsection, all ABPs have width~2, and by a $\weakest$-, $\weak$- or $\gen$-ABP $\Gamma$ we will mean a sequence $\Gamma_k, \ldots, \Gamma_1$ with $\Gamma_k \in \FF^{1\times 2}$, $\Gamma_{k-1}, \ldots, \Gamma_2 \in X^{2\times 2}$, and $\Gamma_1 \in \FF^{2\times 1}$ with $X$ equal to $\HH$, $\SSS$ or $\LL$ respectively. We call $\Gamma_{k-1}, \ldots, \Gamma_2$ the \defin{inner matrices of $\Gamma$}.

\begin{definition}
A matrix $A\in \LL^{2\times 2}$ is called \defin{inherently nondegenerate} (indg) when $\det(A) \in \FF\setminus\{0\}$.
\end{definition}

Allender and Wang prove the following necessary condition for a polynomial to be computable by a $\weakest$-, $\weak$- or $\gen$-ABP whose inner matrices are indg.
Let~$\homog(p)$ denote the highest-degree homogeneous part of a polynomial $p$.

\begin{theorem}[{\cite[Thm.~3.9 and Lem.~4.7]{wang}}]\label{indgcor}
Let $p$ be a polynomial and $\Gamma$ a $\weakest$-, $\weak$- or $\gen$-ABP computing~$p$, whose inner matrices are indg.
Then $\homog(p)$ is a product of affine linear forms.
\end{theorem}

Our next goal is to give a necessary condition for a polynomial $p$ to be computable by a $\weak$-ABP. We begin with a simple lemma,
which can essentially be found in \cite{wang}.

\begin{lemma}[\cite{wang}]\label{zerolem}
Let $p$ be a polynomial.
If $p$ is computed by a $\weak$-ABP that has an inner matrix containing 4 variables, then there is an assignment $\pi$ of 4 variables with $\pi(p) = 0$.
\end{lemma}
\begin{proof}
Let $M$ be such a matrix.
Since the ABP is of type $\weak$, $M$ is of the form
\[
M = \begin{pmatrix}
\alpha_{11} x_{11} + \beta_{11} & \alpha_{12} x_{12} + \beta_{12}\\
\alpha_{21} x_{21} + \beta_{21} & \alpha_{22} x_{22} + \beta_{22}
\end{pmatrix}
\]
for some constants $\alpha_{ij} \in \FF\setminus\{0\}$, $\beta_{ij}\in\FF$.
Applying the four assignments $x_{ij} \mapsto -\beta_{ij}/\alpha_{ij}$ makes~$M$ zero and thus $p$ zero.
\end{proof}

We need two more ideas before we will state and prove the necessary condition we are after. 
(1)~Let $A\in \LL^{2\times 2}$ be nonzero and not-indg (that is, $\det(A)$ is either 0 or a nonconstant polynomial).
Then there is an assignment~$\pi$ of the variables such that $\pi(A)$ has only constant entries and has rank~1.
(2)~Let $p$ be a polynomial computed by an ABP $\Gamma$, that is, $p = \Gamma_k \cdots \Gamma_1$.
Suppose that $\Gamma$ contains a matrix $\Gamma_i$ with only constant entries and with rank~1.
Then there is a $2\times1$ matrix $\Gamma_{i,2}$ and a $1\times 2$ matrix $\Gamma_{i,1}$ such that $\Gamma_i = \Gamma_{i,2} \Gamma_{i,1}$.
Then $p$ is a product
\[
p = p_2 p_1
\]
of polynomials $p_1$, $p_2$, each computable by an ABP, namely
\begin{align*}
p_2 &= \Gamma_k \cdots \Gamma_{i+1} \Gamma_{i,2}\\
p_1 &= \Gamma_{i,1} \Gamma_{i-1} \cdots \Gamma_1.
\end{align*}
We say that $\Gamma_i$ \defin{factors $p$ into $p_2p_1$}. Recall that $\homog(p)$ denotes the highest-degree homogeneous part of a polynomial $p$.
The following is implicit in \cite{wang}.

\begin{theorem}[\cite{wang}]\label{twoabpnec}
Let $p$ be a polynomial computed by a $\weak$-ABP $\Gamma$.
Then there is an assignment $\pi$ of at most 6 variables such that one of the following is true:
\begin{enumerate}
\item $\pi(p)$ is affine linear (including constant), or
\item $\homog(\pi(p))$ is a product of two polynomial of positive degree.
\end{enumerate}
\end{theorem}

\begin{proof}
Let $(\Gamma_k, \ldots, \Gamma_1)$ be the matrices of $\Gamma$, so that $p = \Gamma_k \cdots \Gamma_1$.
If there is a $\Gamma_i$ containing 4 variables, then there is an assignment $\pi$ of these 4 variables with $\pi(p) = 0$ (\protect \MakeUppercase {L}emma\nobreakspace \ref {zerolem}), so we are in case 1.
Otherwise, all $\Gamma_i$ have at most 3 variables. If the inner $\Gamma_i$ are all indg, then $\homog(p)$ is a product of linear forms (\protect \MakeUppercase {T}heorem\nobreakspace \ref {indgcor}), so we are in case 1 or 2. 
Otherwise, there is at least one not-indg inner matrix.
Consider the nonempty subsequence $\mathcal{M} = (M_\ell, \ldots, M_1)$ of not-indg inner matrices. For each $M_i$ there is an assignment $\pi$ of at most 3 variables such that $\pi(M_i)$ has only constant entries and rank 1.
We consider four possible situations.

1. \emph{There is an $M \in \mathcal{M}$ and an assignment $\pi$ of at most~3 variables such that $\pi(M)$ factors $\pi(p)$ into a product of two constants or a product of two polynomials with positive degree.}
Then we are in case 1 or 2.

2. \emph{There is an assignment~$\pi$ of at most 3 variables such that $\pi(M_1)$ factors $\pi(p)$ into~$p_2 p_1$ with $p_2$ a constant and $p_1$ not constant.} Then $p_1$ is computed by an ABP consisting of indg inner matrices (since $M_1$ is the right-most not-indg inner matrix) and hence $\homog(p_1)$ is a product of linear forms (\protect \MakeUppercase {T}heorem\nobreakspace \ref {indgcor}), so we are in case 1 or 2.

3. \emph{There is an assignment~$\pi$ of at most 3 variables such that $\pi(M_\ell)$ factors $\pi(p)$ into~$p_2 p_1$ with $p_2$ not a constant and $p_1$ a constant.}
Then $p_2$ is computed by an ABP consisting of indg inner matrices (since $M_2$ is the left-most not-indg inner matrix) and one proceeds as in the previous situation.

4. \emph{Remaining situation.} In the remaining situation we do the following.
Let $M_i$ be the left-most  matrix in $\mathcal{M}$ such that there is an assignment $\pi$ of at most 3 variables such that $\pi(M_i)$ factors $\pi(p)$ into $p_2 p_1$ with $p_2$ not a constant and $p_1$ constant.
Then there is an assignment $\sigma$ of at most 3 variables such that $\sigma(\pi(M_{i+1}))$ factors $p_2 = p_3 p_4$ with $p_3$ constant and $p_4$ not constant. 
Then $p_4$ is computed by an ABP consisting of indg matrices, and so~$\homog(p_4)$ is a product of homogeneous linear forms.
Therefore we are in case 1 or 2.
\end{proof}

{
\renewcommand{\thetheorem}{\ref{notwabp}}
\begin{theorem}[repeated]
The polynomial
\begin{align*}
p(\xvars) &= (x_{11} + x_{12} + \cdots + x_{17})(x_{21} + x_{22} + \cdots + x_{27})\\ &\phantom{{}={}}{}+ (x_{31} + x_{32} + \cdots + x_{37})(x_{41} + x_{42} + \cdots + x_{47})
\end{align*}
is computable by a width-2 $\gen$-ABP, but not computable by any width-2 $\weak$-ABP.
\end{theorem}
}

\begin{proof} 
Clearly $p(\xvars)$ is computable by a width-2 $\gen$-ABP.
Suppose $p(\xvars)$ is computable by a width-2 $\weak$-ABP.
Then by \protect \MakeUppercase {T}heorem\nobreakspace \ref {twoabpnec} there is an assignment $\pi$ of at most 6 variables such that either $\pi(p)$ is affine linear or $\homog(\pi(p))$ is a product of two polynomials of positive degree.
The first option is impossible, because distinct variables do not cancel.
So $\homog(\pi(p))$ is a product of two polynomials of positive degree.
With another assignment $\sigma$ we can achieve that $\homog(\sigma(\pi(p))$ is of the form $x_i x_j + x_k x_\ell$ for some distinct variables $x_i,x_j,x_k,x_\ell$.
This is not a product of two polynomials of positive degree, so $\homog(\pi(p))$ is not either. 
\end{proof}

\subsection{Comparing different types of edge labels in width-1 ABPs}\label{sec:widthone}

Clearly, $\VPk{1}^\weakest \subseteq \VPk{1}^\weak \subseteq \VPk{1}^\gen$ and $\VPk{1}^* \subseteq \VPk{2}^*$, but this does not give a complete description of all inclusions among these classes.
The following two propositions realize a complete description among $\VPk{1}^*$ and $\VPk{2}^\weakest$.

\begin{proposition}\label{gen1}
$\VPk{1}^\gen \subseteq \VPk{2}^{\weakest}$.
\end{proposition}
\begin{proof}
Let $(p_n) \in \VPk{1}^\gen$. 
Then each $p_n$ is a product of $\polyf(n)$ affine linear forms in $\polyf(n)$ variables. 
Let $\ell(\xvars) = \alpha_0 + \alpha_1 x_1 + \alpha_2 x_2 + \cdots + \alpha_m x_m$ be such an affine linear form with $\alpha_0 \in \FF$ and $\alpha_1, \ldots, \alpha_m \in \FF\setminus\{0\}$.
We can compute $\ell(\xvars)$ with the width-2 $\weakest$-ABP in Fig.\nobreakspace \ref {fig:vp2}.
A product of affine linear forms can be computed by the width-2 $\weakest$-ABP that is the concatenation of the width-2 $\weakest$-ABPs computing the affine linear forms. For $p_n$ the resulting ABP has $\polyf(n)$ size. Thus, $(p_n) \in \VPk{2}^\weakest$.
\end{proof}
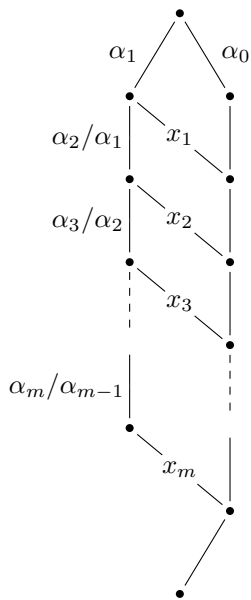
\begin{figure}[h!]
\centering
\begin{tikzpicture}[node distance=1.1cm and 0.6*1.1cm]
\node[state] (21)  {};
\node (22) [right=of 21] {};
\node[state] (31) [below left=of 21] {};
\node[state] (32) [below=of 22] {};
\node[state] (41) [below=of 31] {};
\node[state] (42) [below=of 32] {};
\node[state] (51) [below=of 41] {};
\node[state] (52) [below=of 42] {};
\node (61) [below=of 51] {};
\node[state] (62) [below=of 52] {};
\node (72) [below=of 62] {};
\node[state] (82) [below=of 72] {};
\node[state] (71) [below=of 61] {};
\node (81) [below=of 71] {};
\node[state] (91) [below right=of 81] {};
\node (72) [below=of 62] {};
\path (21) edge node[left, xshift=-0.4em] {$\alpha_1$} (31);
\path (21) edge node[right, xshift=0.4em] {$\alpha_0$} (32);
\path (32) edge (42);
\path (42) edge (52);
\path (52) edge (62);
\path (31) edge node[left] {$\alpha_2/\alpha_1$}(41);
\path (41) edge node[left] {$\alpha_3/\alpha_2$}(51);
\path (51) edge[dashed] (61);
\path (62) edge[dashed] (72);
\path (72) edge (82);
\path (82) edge (91);
\path (31) edge node {$x_1$} (42);
\path (41) edge node {$x_2$} (52);
\path (51) edge node {$x_3$} (62);
\path (61) edge node[left] {$\alpha_m/\alpha_{m-1}$} (71);
\path (71) edge node {$x_m$} (82);
\end{tikzpicture}
\caption{Width-2 $\weakest$-ABP computing $\ell(\xvars) = \alpha_0 + \alpha_1 x_1 + \alpha_2 x_2 + \cdots + \alpha_m x_m$}\label{fig:vp2}
\end{figure}

\begin{proposition}\label{VP1sep}
$\VPk{1}^\weakest \subsetneq \VPk{1}^{\weak} \subsetneq \VPk{1}^\gen \subsetneq \VPk{2}^\weakest$.
\end{proposition}
\begin{proof}
If $(p_n)\in\VPk{1}^{\weakest}$, then $p_n$ is a monomial.
However, $(\alpha_0 + \alpha_1 x_1) \in \VPk{1}^{\weak}$ and $\alpha_0 + \alpha_1x_1$ is not a monomial, so $\VPk{1}^{\weakest}\subsetneq \VPk{1}^{\weak}$.
If $(p_n) \in \VPk{1}^{\weak}$ and $p_n$ is homogeneous, then $p_n$ is a monomial.
However, $\bigl( (x_1 + x_2)^2 \bigr) \in \VPk{1}^\gen$ and $(x_1+x_2)^2$ is not a monomial, so $\VPk{1}^\weak \subsetneq \VPk{1}^\gen$.
The last inclusion is Proposition\nobreakspace \ref {gen1}. To see the strictness, if $(p_n) \in \VPk{1}^\gen$, then the highest-degree homogeneous part $\homog(p_n)$ of $p_n$ is a product of homogeneous linear forms.
However, $(x_1x_2 + x_3x_4) \in \VPk{2}^{\weakest}$ and $x_1x_2 + x_3x_4$ is not a product of homogeneous linear forms, so $\VPk{1}^\gen \subsetneq \VPk{2}^{\weakest}$.
\end{proof}

\subsection{Approximation in width-1 ABPs}

The following proposition says that each of $\VPk{1}^\weakest$, $\VPk{1}^{\weak}$ and $\VPk{1}^\gen$ is closed under approximation.

\begin{proposition}\label{VP1isVP1bar}
$\VPk{1}^* = \approxbar{\VPk{1}^*}$.
\end{proposition}
\begin{proof}
Trivially, $\VPk{1}^*\subseteq \approxbar{\VPk{1}^*}$.
To prove the opposite inclusion, let $(f_n) \in \approxbar{\VPk{1}^*}$.
There are polynomials $g_n(\eps, \xvars) \in \FF[\eps, \xvars]$ such that $f_n + \eps g_n(\eps,\xvars)$ can be written as a product of $\polyf(n)$ affine linear forms in $\FF(\eps)[\xvars]$ in $\polyf(n)$ variables (these affine linear forms have either $\weakest$-, $\weak$- or $\gen$-type).
That is, (forgetting the subscript $n$ for the moment) $f(\xvars) + \eps g(\eps, \xvars)$ can be written as
\[
f(\xvars) + \eps g(\eps, \xvars) = \prod_{\smash{i=1}}^m \ell_i(\eps, \xvars)
\]
with
\[
\ell_i(\eps, \xvars) = \sum_{\mathclap{j=d_i}}^{\smash{e_i}} \eps^j k_{i,j}(\xvars)
\]
for some affine linear forms $k_{i,j} \in \FF[\xvars]$, such that $k_{i,d_i}(\xvars)\neq0$, and $d_i\leq e_i \in \ZZ$.
By shifting $\eps$-factors from $\ell_1, \ldots, \ell_{m-1}$ to $\ell_m$ we can assume that $d_i = 0$ for $i< m$. We claim that $d_m \geq 0$. If $d_m < 0$, then 
expanding $\prod_i \ell_i(\xvars)$ as a Laurent series in $\eps$ gives a term with a negative power of $\eps$.
This contradicts $f(\xvars) + \eps g(\xvars)$ having only nonnegative powers of $\eps$.
Therefore, the $\ell_i(\xvars)$ do not contain any negative powers of $\eps$ and we can safely substitute $\eps \mapsto 0$ in each linear form $\ell_i$ to obtain $f$ as a product of affine linear forms in $\FF[\xvars]$ (either of $\weakest$-, $\weak$- or $\gen$-type).
Remembering our subscript $n$ again, we have thus proven $(f_n) \in \VPk{1}^{*}$.
\end{proof}

\subsection{Nondeterminism in width-1 ABPs}

In the following proposition we compare $\VPk{1}^*$ to $\VNPk{1}^*$ for all three versions $*\in\{\weakest,\weak,\gen\}$.

\begin{proposition}\label{VP1isVNP1weakest}\hfil
\begin{itemize}
\item $\VPk{1}^* = \VNP_1^*$ for $*$ equal to $\weakest$ or $\weak$.
\item $\VPk{1}^\gen \subsetneq \VNP_1^\gen$ when $\characteristic(\FF)\neq2$. 
\end{itemize}
\end{proposition}
\begin{proof}
Trivially, $\VPk{1}^* \subseteq \VNPk{1}^*$. Let $(p_n) \in \VNPk{1}^\weakest$.
Then $p_n$ can be written as a hypercube-sum over a monomial,
\[
p(\xvars) = \sum_{\mathclap{\bvars \in \{0,1\}^{\polyf(n)}}} m(\bvars, \xvars)
\]
with $m$ a monomial (subscripts $n$ are implied). 
For any $\bvars$-variable that does not occur in $m$, we remove that $\bvars$-variable form the summation and at the same time multiply the expression by 2, to again have an expression for $p(\xvars)$.
Assuming all $\bvars$-variables occur in $m$, only for $\bvars=(1,1,\ldots,1)$ can $m(\bvars, \xvars)$ be nonzero. 
So $p(\xvars) = m((1,\ldots,1), \xvars)$. 
Remembering the subscript $n$, we proved $(p_n) \in \VPk{1}^{\weakest}$.

Let $(p_n) \in \VNPk{1}^{\weak}$. 
Then, (forgetting the subscript $n$)
\[
p(\xvars) = \sum_{\bvars \in \{0,1\}^{\polyf(n)}} \prod_i \ell_i(\bvars) \prod_j k_j(\xvars)
\]
for some simple affine linear forms $\ell_i$ in the variables $\bvars$ and some simple affine linear forms~$k_j$ in the variables $\xvars$. 
The product $\prod_j k_j(\xvars)$ is independent of $\bvars$, while $\sum_{\bvars} \prod_i \ell_i(\bvars)$ is a constant. 
We can thus write $p(\xvars)$ as a constant times $\prod_j k_j(\xvars)$.
Therefore (remembering $n$), $p_n(\xvars) \in \VPk{1}^\weak$.
This proves the first line of the proposition.

To prove the second line, recall that if $(p_n) \in \VPk{1}^\gen$, then $p_n$ is a product of affine linear forms.
However, let $p_n(x_1, x_2) = \sum_{b\in \{0,1\}} (x_1 + b)(x_2 + b) = 2x_1x_2 + x_1 + x_2 + 1$. Then $(p_n) \in \VNPk{1}^\gen$, but $p_n(x_1, x_2)$ is a not a product of affine linear forms, as we will now verify.
Suppose $2x_1x_2 + x_1 + x_2 + 1 = (\alpha_0 + \alpha_1x_1 + \alpha_2 x_2)(\beta_0 + \beta_1 x_1 + \beta_2 x_3)$.
Then $\alpha_1\beta_1=0$ and $\alpha_2 \beta_2 = 0$.
Since $\alpha_1 \beta_1 = 0$,
we may assume without loss of generality that $\alpha_1 = 0$.
Since not both $\alpha_1$ and $\alpha_2$ can be 0 (otherwise $(\alpha_0 + \alpha_1 x_1 + \alpha_2 x_2)(\beta_0 + \beta_1 x_1 + \beta_2 x_2)$ has degree 1) and since $\alpha_2\beta_2 = 0$, we have $\beta_2 = 0$. Hence, $2 x_1 x_2 + x_1 + x_2 + 1 = (\alpha_0 + \alpha_2 x_2)(\beta_0 + \beta_1 x_1)$. Then $\alpha_0 \beta_0 = 1$,  $\alpha_0 \beta_1 = 1$, $\alpha_2 \beta_0 = 1$, and $\alpha_2 \beta_1 = 2$. The first two of these equations imply $\beta_0=\beta_1$, which contradicts the last two of these equations.
So $\VPk{1}^\gen \subsetneq \VNPk{1}^\gen$.
\end{proof}

\begin{remark}\label{vnpstrict}
It follows directly from Proposition\nobreakspace \ref {VP1isVNP1weakest} and Proposition\nobreakspace \ref {VP1sep} that we have strict inclusions
$\VNP_1^\weakest \subsetneq \VNP_1^\weak \subsetneq \VNP_1^\gen$, when $\characteristic(\FF)\neq 2$.
\end{remark}

\section{Alternative proof of $\VNPk{1} = \VNP$ via $\VPk{3}$}\label{sec:VNPdirect}

Recall that in Section\nobreakspace \ref {sec:secondmain} we proved that
\begin{equation}\label{eq:vnp1vnp}
\VNPk{1}^\gen = \VNP
\end{equation}
using the completeness of the permanent (\protect \MakeUppercase {T}heorem\nobreakspace \ref {VNPmain}). We will present an alternative proof of~\eqref{eq:vnp1vnp} inspired by the proof of the following theorem by Ben-Or and Cleve. The alternative proof of~\eqref{eq:vnp1vnp} has the benefit that it can be extended to show a slightly stronger result, see \protect \MakeUppercase {T}heorem\nobreakspace \ref {thm:secondmainstronger}. 

\begin{theorem}[Ben-Or and Cleve, \cite{ben1992computing}]\label{benorcleve}
For $k\geq 3$, $\VPk{k}^* = \VPe$.
\end{theorem}

\begin{proof} 
Proposition\nobreakspace \ref {prop:vpkinvpe} says that $\VPk{k}^* \subseteq \VPe$.
We will prove that $\VPe \subseteq \VPk{3}^\weakest$, from which it follows that $\VPe \subseteq \VPk{k}^*$ and thus $\VPk{k}^* = \VPe$. For a polynomial~$h$, define the matrix
\[
M(h) \coloneqq \begin{pmatrix}
1 & 0 & 0\\
h & 1 & 0\\
0 & 0 & 1\\
\end{pmatrix}
\]
which, as part of an ABP, looks like
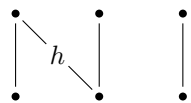
\[
\begin {tikzpicture}
\node[state] (11) {};
\node[state] (12) [right=of 11] {};
\node[state] (13) [right=of 12] {};
\node[state] (21) [below=of 11] {};
\node[state] (22) [right=of 21] {};
\node[state] (23) [below=of 13] {};
\path (11) edge node {$h$} (22);
\path (12) edge (22);
\path (13) edge (23);
\path (11) edge (21);
\end{tikzpicture}
\]
We call the following matrices \defin{primitive}:
\begin{itemize}
\item $M(h)$ with $h$ any variable or any constant in $\FF$
\item every $3\times 3$ permutation matrix $M_\pi$ with $\pi \in S_3$ any permutation
\item every diagonal matrix $M_{a,b,c} \coloneqq \diag(a,b,c)$ with $a,b,c$ any constants in $\FF$
\end{itemize}
The entries of the primitives are variables or constants in $\FF$, making them suitable to use in the construction of a width-3 $\weakest$-ABP (\protect \MakeUppercase {D}efinition\nobreakspace \ref {modelsubsets}).

Let $(f_n) \in \VPe$. Then $f_n$ can be computed by a formula of size $s(n) \in \poly(n)$.
By Brent's depth-reduction theorem for formulas (\cite{brent1974parallel}) $f_n$ can then also be computed by a formula of size $\poly(n)$ and depth $d(n) \in \Oh(\log n)$.

We will construct a sequence of primitives $A_1, \ldots, A_{m(n)}$ such that
\[
A_1 \cdots A_{m(n)} =
\begin{pmatrix}
1 & 0 & 0\\
f_n & 1 & 0\\
0 & 0 & 1
\end{pmatrix}
\]
with $m(n) \in \Oh(4^{d(n)}) = \poly(n)$. Then
\[
f_n(\xvars) = \begin{psmallmatrix}1 & 1 & 1\end{psmallmatrix} M_{-1,1,0} A_1 \cdots A_m \begin{psmallmatrix}1 \\ 1 \\ 1\end{psmallmatrix},
\]
so $f_n(\xvars)$ can be computed by a width-3 $\weakest$-ABP of size $\poly(n)$, proving the theorem.

To explain the construction, let $h$ be a polynomial and consider a formula computing~$h$ of depth $d$. The goal is to construct (recursively on the formula structure) primitives $A_1, \ldots, A_m$ such that
\begin{equation}\label{eq:goal}
A_1\cdots A_m = 
\begin{pmatrix}
1 & 0 & 0\\
h & 1 & 0\\
0 & 0 & 1
\end{pmatrix}
\quad \textnormal{with $m \in \Oh(4^d)$.}
\end{equation}

Suppose $h$ is a variable or a constant. Then $M(h)$ is itself a primitive matrix.

Suppose $h = f+g$ is a sum of two polynomials $f,g$ and suppose $M(f)$ and $M(g)$ can be written as a product of primitives. Then $M(f+g)$ equals a product of primitives, because $M(f+g) = M(f)M(g)$. This can easily be verified directly, or by noting that in the corresponding partial ABPs the top-bottom paths ($u_i$-$v_j$ paths) have the same value:
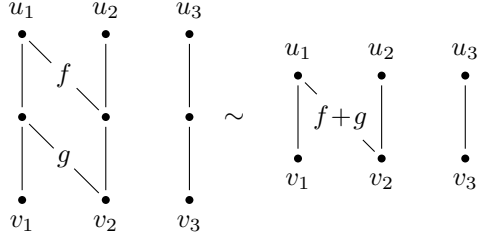
\[
\begin{minipage}{7em}
\begin {tikzpicture}
\node[state] (11) [label={$u_1$}]{};
\node[state] (12) [right=of 11, label={$u_2$}] {};
\node[state] (13) [right=of 12, label={$u_3$}] {};

\node[state] (21) [below=of 11] {};
\node[state] (22) [right=of 21] {};
\node[state] (23) [below=of 13] {};

\node[state] (31) [below=of 21, label={-90:$v_1$}] {};
\node[state] (32) [below=of 22, label={-90:$v_2$}] {};
\node[state] (33) [below=of 23, label={-90:$v_3$}] {};

\path (11) edge node {$f$} (22);
\path (12) edge (22);
\path (13) edge (23);
\path (23) edge (33);
\path (21) edge node {$g$} (32);

\path (11) edge (21);
\path (21) edge (31);
\path (22) edge (32);
\end{tikzpicture}
\end{minipage}
\quad\sim\quad
\begin{minipage}{7em}
\begin {tikzpicture}
\node[state] (11) [label={$u_1$}]{};
\node[state] (12) [right=of 11, label={$u_2$}] {};
\node[state] (13) [right=of 12, label={$u_3$}] {};

\node[state] (21) [below=of 11, label={-90:$v_1$}] {};
\node[state] (22) [below=of 12, label={-90:$v_2$}] {};
\node[state] (23) [below=of 13, label={-90:$v_3$}] {};
\path (11) edge node {$f\!+\!g$} (22);
\path (12) edge (22);
\path (13) edge (23);
\path (11) edge (21);
\end{tikzpicture}
\end{minipage}
\]

Suppose $h = fg$ is a product of two polynomials $f,g$ and suppose $M(f)$ and $M(g)$ can be written as a product of primitives. Then $M(fg)$ equals a product of primitives, because
\[
M(f\cdot g) = M_{(23)}  \bigl(M_{1,-1,1} M_{(123)} M(g) M_{(132)} M(f)\bigr)^2 M_{(23)}
\]
(here $(23)\in S_3$ denotes the transposition $1\mapsto1,2\mapsto3,3\mapsto2$ and $(123)\in S_3$ denotes the cyclic shift $1\mapsto 2, 2\mapsto3, 3\mapsto1$) as can be verified either directly or by checking that in the corresponding partial ABPs the top-bottom paths ($u_i$-$v_j$ paths) have the same value:
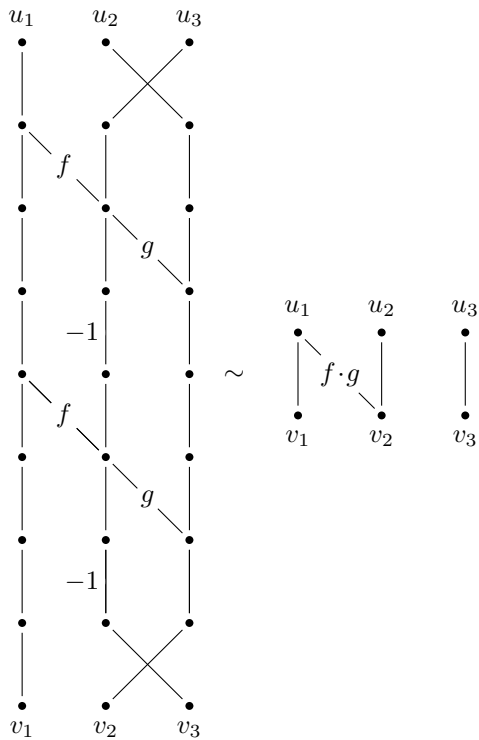
\[
\begin{minipage}{7em}
\begin {tikzpicture}
\node[state] (01) [label={$u_1$}] {};
\node[state] (02) [right=of 01, label={$u_2$}] {};
\node[state] (03) [right=of 02, label={$u_3$}] {};
\node[state] (11) [below=of 01] {};
\node[state] (12) [below=of 02] {};
\node[state] (13) [below=of 03] {};
\node[state] (21) [below=of 11] {};
\node[state] (22) [below=of 12] {};
\node[state] (23) [below=of 13] {};
\node[state] (31) [below=of 21] {};
\node[state] (32) [below=of 22] {};
\node[state] (33) [below=of 23] {};
\node[state] (41) [below=of 31] {};
\node[state] (42) [below=of 32] {};
\node[state] (43) [below=of 33] {};
\node[state] (51) [below=of 41] {};
\node[state] (52) [below=of 42] {};
\node[state] (53) [below=of 43] {};
\node[state] (61) [below=of 51] {};
\node[state] (62) [below=of 52] {};
\node[state] (63) [below=of 53] {};
\node[state] (71) [below=of 61] {};
\node[state] (72) [below=of 62] {};
\node[state] (73) [below=of 63] {};
\node[state] (81) [below=of 71, label={-90:$v_1$}] {};
\node[state] (82) [below=of 72, label={-90:$v_2$}] {};
\node[state] (83) [below=of 73, label={-90:$v_3$}] {};
\path (01) edge (11);
\path (02) edge (13);
\path (03) edge (12);
\path (11) edge node {$f$} (22);
\path (12) edge (22);
\path (13) edge (23);
\path (23) edge (33);
\path (22) edge (32);
\path (32) edge node[left,overlay] {$-1$} (42);
\path (33) edge (43);
\path (22) edge node {$g$} (33);
\path (11) edge (21);
\path (21) edge (31);
\path (31) edge (41);
\path (41) edge (52);
\path (62) edge (72);
\path (41) edge (51);
\path (41) edge node {$f$} (52);
\path (42) edge (52);
\path (43) edge (53);
\path (51) edge (61);
\path (52) edge (62);
\path (52) edge node {$g$} (63);
\path (61) edge (71);
\path (62) edge node[left] {$-1$} (72);
\path (63) edge (73);
\path (53) edge (63);
\path (63) edge (73);
\path (71) edge (81);
\path (72) edge (83);
\path (73) edge (82);
\end{tikzpicture}
\end{minipage}
\quad\sim\quad
\begin{minipage}{7em}
\begin {tikzpicture}
\node[state] (11) [label={$u_1$}]{};
\node[state] (12) [right=of 11, label={$u_2$}] {};
\node[state] (13) [right=of 12, label={$u_3$}] {};
\node[state] (21) [below=of 11, label={-90:$v_1$}] {};
\node[state] (22) [below=of 12, label={-90:$v_2$}] {};
\node[state] (23) [below=of 13, label={-90:$v_3$}] {};
\path (11) edge node {$f\!\cdot\! g$} (22);
\path (12) edge (22);
\path (13) edge (23);
\path (11) edge (21);
\end{tikzpicture}
\end{minipage}
\]
This completes the construction.

The length $m$ of the construction is $m(h) = 1$ for $h$ a variable or constant and recursively $m(f+g) = m(f) + m(g)$, $m(f\cdot g) = 2(m(f) + m(g)) + \Oh(1)$, so $m\in \Oh(4^d)$ where $d$ is the formula depth of $h$. The construction thus satisfies \eqref{eq:goal}, proving the theorem.
\end{proof}

We will now give an alternative proof of \protect \MakeUppercase {T}heorem\nobreakspace \ref {VNPmain}.
{
\renewcommand{\thetheorem}{\ref{VNPmain}}
\begin{theorem}[repeated]
$\VNP_{1} = \VNP$ when $\characteristic(\FF) \neq 2$.
\end{theorem}
}

\begin{proof}\label{altproof}
Clearly, $\VNPk{1}^\gen \subseteq \VNP$ by Proposition\nobreakspace \ref {prop:vpkinvpe} and taking the nondeterminism closure~$\N$. We will prove that $\VNP\subseteq \VNPk{1}^\gen$.

Recall that in the proof  of $\VPe \subseteq \VPk{3}^\weakest$ (\protect \MakeUppercase {T}heorem\nobreakspace \ref {benorcleve}), we defined for any polynomial~$h$ the matrix 
\[
M(h) \coloneqq \begin{pmatrix}
1 & 0 & 0 \\
h & 1 & 0 \\
0 & 0 & 1 \\
\end{pmatrix}
\]
and we called the following matrices \emph{primitive}:
\begin{itemize}
\item $M(h)$ with $h$ any variable or any constant in $\FF$
\item every $3\times 3$ permutation matrix $M_\pi$ with $\pi \in S_3$ any permutation
\item every diagonal matrix $M_{a,b,c} \coloneqq \diag(a,b,c)$ with $a,b,c$ any constants.
\end{itemize}
In the proof of $\VPe \subseteq \VPk{3}^\weakest$ we constructed, for any family $(f_n) \in \VPe$ a sequence of primitives $A_{n,1}, \ldots, A_{n,t(n)}$ with $t(n) \in \polyf(n)$ such that 
\[
f_n(\xvars) = \begin{psmallmatrix}1 & 1 & 1\end{psmallmatrix} M_{-1,1,0} A_1 \cdots A_m \begin{psmallmatrix}1 \\ 1 \\ 1\end{psmallmatrix}.
\]
We will construct a hypercube sum over a width-1 $\gen$-ABP that evaluates the right-hand side, to show that $\VPe \subseteq \VNPk{1}^\gen$. This implies $\VNPe \subseteq \VNPk{1}^\gen$. Then by Valiant's \protect \MakeUppercase {T}heorem\nobreakspace \ref {valiant}, $\VNP\subseteq \VNPk{1}^\gen$.

Let $f(\xvars)$ be a polynomial and let $A_1, \ldots, A_k$ be primitives such that $f(\xvars)$ is computed as 
\[
f(\xvars) = \begin{psmallmatrix}1 & 1 & 1\end{psmallmatrix}  A_k \cdots A_1 \begin{psmallmatrix}1 \\ 1 \\ 1\end{psmallmatrix}.
\]
View this expression as a width-3 ABP $G$, with vertex layers labeled as shown in the left diagram of Fig.\nobreakspace \ref {vnpfig}.
\begin{figure}
\centering
\begin{minipage}{18em}
\begin {tikzpicture}
\node (01) {};
\node[state] (02) [right=of 01, label={$s$}] {};
\node (03) [right=of 02] {};

\node[state] (11) [below=of 01] {};
\node[state] (12) [below=of 02] {};
\node[state] (13) [below=of 03] {};
\node (0) [left=of 11] {$0$};

\node[state] (21) [below=of 11] {};
\node[state] (22) [below=of 12] {};
\node[state] (23) [below=of 13] {};
\node (1) [left=of 21] {$1$};

\node[state] (31) [below=of 21] {};
\node[state] (32) [below=of 22] {};
\node[state] (33) [below=of 23] {};
\node (2) [left=of 31] {$2$};

\node[state] (41) [below=of 31] {};
\node[state] (42) [below=of 32] {};
\node[state] (43) [below=of 33] {};
\node (3) [left=of 41] {$k\compactminus1$};

\node[state] (51) [below=of 41] {};
\node[state] (52) [below=of 42] {};
\node[state] (53) [below=of 43] {};
\node (4) [left=of 51] {$k$};

\node (61) [below=of 51] {};
\node[state] (62) [below=of 52, label={-90:$t$}] {};
\node (63) [below=of 53] {};

\path (02) edge (11);
\path (02) edge (13);
\path (02) edge (12);

\path (11) edge (21);
\path (11) edge (22);
\path (11) edge (23);
\path (12) edge (21);
\path (12) edge (22);
\path (12) edge (23);
\path (13) edge (21);
\path (13) edge (22);
\path (13) edge node[right=1em]{$A_1$} (23);

\path (21) edge (31);
\path (21) edge (32);
\path (21) edge (33);
\path (22) edge (31);
\path (22) edge (32);
\path (22) edge (33);
\path (23) edge (31);
\path (23) edge (32);
\path (23) edge node[right=1em]{$A_2$} (33);

\path (31) edge[dashed] (41);
\path (32) edge[dashed] (42);
\path (33) edge[dashed] (43);

\path (41) edge (51);
\path (41) edge (52);
\path (41) edge (53);
\path (42) edge (51);
\path (42) edge (52);
\path (42) edge (53);
\path (43) edge (51);
\path (43) edge (52);
\path (43) edge node[right=1em]{$A_k$} (53);

\path (51) edge (62);
\path (52) edge (62);
\path (53) edge (62);
\end{tikzpicture}
\end{minipage}
\begin{minipage}{10em}
\begin {tikzpicture}
\node (01) {};
\node[state] (02) [right=of 01, label={$s$}] {};
\node (03) [right=of 02] {};

\node[state] (11) [below=of 01] {};
\node[state] (12) [below=of 02] {};
\node[state] (13) [below=of 03] {};
\node (0) [right=of 13] {1\,0\,0};

\node[state] (21) [below=of 11] {};
\node[state] (22) [below=of 12] {};
\node[state] (23) [below=of 13] {};
\node (1) [right=of 23] {0\,1\,0};

\node[state] (31) [below=of 21] {};
\node[state] (32) [below=of 22] {};
\node[state] (33) [below=of 23] {};
\node (2) [right=of 33] {0\,1\,0};

\node[state] (41) [below=of 31] {};
\node[state] (42) [below=of 32] {};
\node[state] (43) [below=of 33] {};
\node (3) [right=of 43] {0\,0\,1};

\node[state] (51) [below=of 41] {};
\node[state] (52) [below=of 42] {};
\node[state] (53) [below=of 43] {};
\node (4) [right=of 53] {0\,1\,0};

\node (61) [below=of 51] {};
\node[state] (62) [below=of 52, label={-90:$t$}] {};
\node (63) [below=of 53] {};

\path (02) edge[gray] (13);
\path (02) edge[gray] (12);
\path (02) edge[line width=1.5pt, shorten <=-1pt, shorten >= -1pt] (11);

\path (11) edge[gray] (21);
\path (11) edge[gray] (23);
\path (12) edge[gray] (21);
\path (12) edge[gray] (22);
\path (12) edge[gray] (23);
\path (13) edge[gray] (21);
\path (13) edge[gray] (22);
\path (13) edge[gray] (23);
\path (11) edge[line width=1.5pt, shorten <=-1pt, shorten >= -1pt] (22);

\path (21) edge[gray] (31);
\path (21) edge[gray] (32);
\path (21) edge[gray] (33);
\path (22) edge[gray] (31);
\path (22) edge[gray] (33);
\path (23) edge[gray] (31);
\path (23) edge[gray] (32);
\path (23) edge[gray] (33);
\path (22) edge[line width=1.5pt, shorten <=-1pt, shorten >= -1pt] (32);

\path (31) edge[dashed, gray] (41);
\path (32) edge[dashed, gray] (42);
\path (33) edge[dashed, gray] (43);

\path (41) edge[gray] (51);
\path (41) edge[gray] (52);
\path (41) edge[gray] (53);
\path (42) edge[gray] (51);
\path (42) edge[gray] (52);
\path (42) edge[gray] (53);
\path (43) edge[gray] (53);
\path (43) edge[line width=1.5pt, shorten <=-1pt, shorten >= -1pt] (52);

\path (51) edge[gray] (62);
\path (53) edge[gray] (62);
\path (52) edge[line width=1.5pt, shorten <=-1pt, shorten >= -1pt] (62);
\end{tikzpicture}
\end{minipage}
\caption{Illustration of layer labelling and path labelling in the proof of \protect \MakeUppercase {T}heorem\nobreakspace \ref {VNPmain}.}\label{vnpfig}
\end{figure}
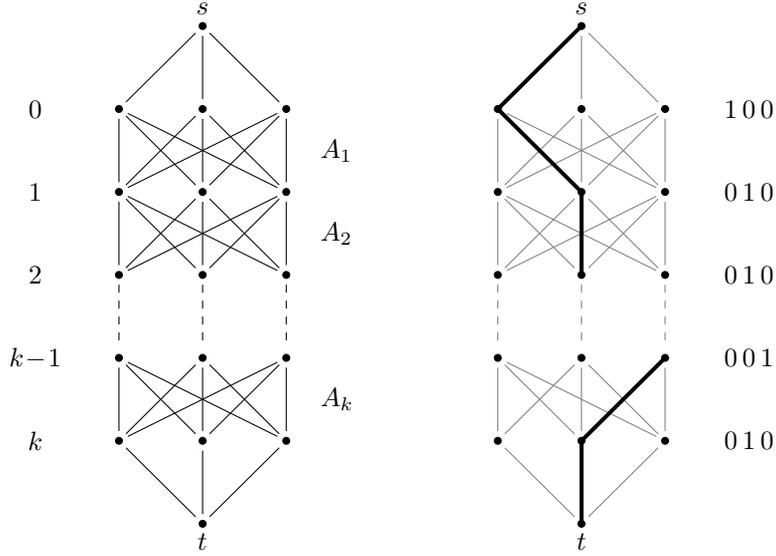
Assume for simplicity that all edges between layers are present, possibly with label 0. The sum of the values of every $s$-$t$ path in $G$ equals $f(\xvars)$,
\begin{equation}\label{eq:paths}
f(\xvars) = \sum_{\mathclap{j \in [3]^k}} A_k[j_k,j_{k-1}] \cdots A_1[j_2, j_1].
\end{equation}

We now introduce some hypercube variables. To every vertex, except $s$ and $t$, we associate a bit; the bits in the $i$th layer we call $b_1[i]$, $b_2[i]$, $b_3[i]$. To an $s$-$t$ path in $G$ we associate an assignment of the $b_j[i]$ by setting the bits of vertices visited by the path to 1 and the others to 0. For example, in the right diagram in Fig.\nobreakspace \ref {vnpfig} we show an $s$-$t$ path with the corresponding assignment of the bits $b_1[i]$, $b_2[i]$, and $b_3[i]$.
The assignments of $b_j[i]$ corresponding to $s$-$t$ paths are the ones such that for every $i\in [k]$ exactly one of $b_1[i]$, $b_2[i]$, $b_3[i]$ equals 1. Let
\begin{equation}\label{Vdef}
V(b_1, b_2, b_3) \coloneqq \prod_{i\in [k]} (b_1[i]+b_2[i]+b_3[i]) \prod_{\mathclap{\substack{s,t\in [3]:\\s\neq t}}} \bigl(1-b_s[i] b_t[i]\bigr).
\end{equation}
The assignments of $b_j[i]$ corresponding to $s$-$t$ paths are thus the ones such that $V(b_1,b_2,b_3) = 1$. Otherwise, $V(b_1, b_2, b_3) = 0$.

We will now write $f(\xvars)$ as a hypercube sum by replacing each $A_i[j_i, j_{i-1}]$ in \eqref{eq:paths} by a product of affine linear forms $S_i(A_i)$ with variables $\bvars$ and $\xvars$ as follows
\[
\sum_{\bvars} V(b_1, b_2, b_3) S_k(A_k) \cdots S_1(A_1).
\]
Define $\Eq(\alpha, \beta) : \{0,1\}^2 \to \{0,1\}$ by $(1-\alpha-\beta)(1-\alpha-\beta)$. This function is 1 if $\alpha=\beta$ and~0 otherwise.

\begin{itemize}
\item
For any variable or constant $x$ define
\begin{align*}
\Sop_i(M(x)) &\coloneqq \bigl(1+ (x-1)(b_1[i] - b_1[i\compactminus 1])\bigr)\\
&\quad \cdot \bigl(1-(1-b_2[i]) b_2[i\compactminus 1]\bigr)\\
&\quad \cdot \Eq\bigl(b_3[i\compactminus 1], b_3[i]\bigr).
\end{align*}
\item
For any permutation $\pi \in S_3$ define
\begin{align*}
\Sop_i(M_\pi) &\coloneqq \Eq\bigl(b_1[i\compactminus 1], b_{\pi(1)}[i]\bigr)\\
&\quad\cdot\Eq\bigl(b_2[i\compactminus 1], b_{\pi(2)}[i]\bigr)\\ &\quad\cdot\Eq\bigl(b_3[i\compactminus1], b_{\pi(3)}[i]\bigr).
\end{align*}
\item 
For any constants $a,b,c \in \FF$ define
\begin{align*}
\Sop_i(M_{a,b,c})  &\coloneqq \bigl( a\cdot b_1[i\compactminus1] + b\cdot b_2[i\compactminus1] + c\cdot b_3[i-1]\bigr)\\
&\quad \cdot\Eq\bigl(b_1[i\compactminus1], b_1[i]\bigr)\\
&\quad \cdot\Eq\bigl(b_2[i\compactminus1], b_2[i]\bigr)\\ 
&\quad \cdot\Eq\bigl(b_3[i\compactminus1], b_3[i]\bigr).
\end{align*}
\end{itemize}
One verifies that with these definitions indeed
\[
f(\xvars) = \sum_{\bvars} V(b_1, b_2, b_3) S_k(A_k) \cdots S_1(A_1).
\]
Some of the factors in the $S_i(A_i)$ are not affine linear. As a final step we apply the equation
$1 + xy = \tfrac{1}{2} \sum_{c\in \{0,1\}} (x + 1-2c)(y+1-2c)$ (\protect \MakeUppercase {L}emma\nobreakspace \ref {tricklem})
to write these factors as products of affine linear forms, introducing new hypercube variables.
\end{proof}

Combining \protect \MakeUppercase {T}heorem\nobreakspace \ref {VNPmain} and \protect \MakeUppercase {R}emark\nobreakspace \ref {vnpstrict} gives the separation $\VNP_1^\weak \subsetneq \VNP_1^\gen = \VNP$. 
We can prove a slightly stronger separation by adjusting the construction in the above proof of \protect \MakeUppercase {T}heorem\nobreakspace \ref {VNPmain}. 
Namely, let $\SSS^+ \coloneqq \{\alpha x_i + \beta x_j + \gamma \mid \alpha, \beta, \gamma \in \FF\}$ be the set of affine linear forms in at most two variables and let $\VPk{1}^{\weak+}$ be the class of families that can be computed by width\nobreakdash-1 ABPs over $\mathbb{S}^+$ of polynomial size. 
Define $\VNP_1^{\weak+}$ accordingly (\protect \MakeUppercase {D}efinition\nobreakspace \ref {Ndef}). 
Then we can adjust the construction in the above proof of \protect \MakeUppercase {T}heorem\nobreakspace \ref {VNPmain} to show the following.

\begin{theorem}\label{thm:secondmainstronger}
$\VNP_1^\weak \subsetneq \VNP_1^{\weak+} = \VNP$ when $\characteristic(\FF) \neq 2$.
\end{theorem}
\begin{proof}
We only need to show $\VNPk{1}^{\weak+} = \VNP$, as $\VNPk{1}^\weak \subsetneq \VNP$ was shown in \protect \MakeUppercase {R}emark\nobreakspace \ref {vnpstrict}.
The adjustments we have to make to the construction in the proof  
of \protect \MakeUppercase {T}heorem\nobreakspace \ref {VNPmain} are as follows.
Most of the resulting polynomial of the construction is already of the correct form where each linear forms contains at most two variables, since the expression $\Eq(x,y) = (1-x-y)^2$ and the expression $1 + xy = \tfrac{1}{2} \sum_{c\in \{0,1\}} (x + 1-2c)(y+1-2c)$ are of this form.
Three expressions occur that are not of the correct form:
\begin{enumerate}
\item $b_1[i] + b_2[i] + b_3[i]$ \quad in $V(b_1, b_2, b_3)$,
\item $a\cdot b_1[i\compactminus1] + b\cdot b_2[i\compactminus1] + c\cdot b_3[i-1]$\quad in $\Sop(M_{a,b,c})$, and
\item $1 + (x-1)(b_1[i] - b_1[i\compactminus1])$\quad in $\Sop(M(x))$
\end{enumerate}
Expression 1 and expression 2 we can write in the correct form using the identity
\begin{equation}\label{tripletrick}
\tfrac{1}{2} \sum_{\mathclap{b\in \{0,1\}}} (x + 1-2b)(y + 1-2b)(z + 1-2b) = x + y + z + xyz.
\end{equation}
Indeed, expression 1 can be replaced by
\begin{align*}
&\tfrac{1}{2} \sum_{\mathclap{c\in \{0,1\}}} (b_1[i] + 1-2c)(b_2[i] + 1-2c)(b_3[i] + 1-2c)\\ &\quad= b_1[i] + b_2[i] + b_3[i] + b_1[i]b_2[i]b_3[i],
\end{align*}
since the unwanted term $b_1[i]b_2[i]b_3[i]$ will always vanish in our construction (because in \eqref{Vdef} we multiply with $1-b_s[i] b_t[i]$ for every $s\neq t$). Similarly for expression 2.

For expression 3, we first replace the expression $1+(x-1)(b_1[i] - b_1[i\compactminus1])$ by the expression $\tfrac{1}{2} \sum_{c\in \{0,1\}} (x-1 + 1-2c)(b_1[i] - b_1[i\compactminus 1] + 1 - 2c)$. The second factor has too many variables. We replace it, using identity \eqref{tripletrick}, by 
\begin{align*}
&\tfrac{1}{2} \sum_{\mathclap{c'\in \{0,1\}}} \bigl(b_1[i] + 1-2c'\bigr)\bigl(- b_1[i\compactminus1] + 1 + 1-2c' \bigr)\bigl(-2c + 1-2c'\bigr)\\
&\quad = b_1[i] - b_1[i\compactminus1] + 1 - 2c +  b_1[i]\bigl(1-b_1[i\compactminus1]\bigr)\bigl(-2c\bigr).
\end{align*}
The first four summands in the right-hand side are as we want. The last summand is only nonzero if $b_1[i] = 1$ and $b_1[i\compactminus1] = 0$. However, since $\Sop_i(M(x))$ contains a factor $1-(1-b_2[i])b_2[i\compactminus1]$ and a factor $\Eq(b_3[i\compactminus1],b_3[i])$, it can be checked that this last summand will always vanish.

In the new construction thus obtained each linear form is in $\SSS^+$. This completes the necessary adjustments to the construction.
\end{proof}

\section{Constant-width ABPs have small formulas}
The following well-known proposition says that the iterated product of constant-size matrices can be efficiently computed by a formula.

\begin{proposition}\label{prop:vpkinvpe}
Let $k\geq 1$. Then $\VPk{k} \subseteq \VPe$.
\end{proposition}

\begin{proof}
Let $(f_n) \in \VPk{k}$, so $f_n$ has $v(n) \in \polyf(n)$ variables.
There is a function $m(n)\in \polyf(n)$ and there are $k\times k$ matrices $M_{n,1}, \ldots, M_{n,m(n)}$ with affine linear forms as entries, such that $\tr(M_{n,1}\cdots M_{n,m(n)}) = f_n$.
We may assume that each affine linear form occurring in $M_{n,i}$ has $v(n)$ variables, since $f_n$ has $v(n)$ variables.
We will recursively construct a multi-output formula computing the product $M_{n,1}\cdots M_{n,m(n)}$.
From this one can efficiently compute the trace. Let $\size(m,n)$ denote the size of the formula that we construct. 
Let $w(n):=2v(n)$.
A single matrix $M_{n,i}$ we can compute by a multi-output formula of size~$k^2 w(n)$ (the $w(n)$ is needed to compute each affine linear form). So $\size(1,n) = k^2 w(n)$. 
Suppose that matrices $A$ and $B$ can  be computed by a multi-output formulas $F$ and $G$ respectively, each of size $s$. 
Then the product $AB$ can be computed by a multi-output formula of size $2k s + c(k)$ with $c(k)\in \Oh(k^3)$, as follows: take $k$ copies of the formula $F$ for $A$ and take~$k$ copies of the formula~$G$ for $B$ and appropriately add $c(k)$ $\times$- and $+$-gates in order to perform the matrix multiplication. 
(The reason that we take $k$ copies of the formulas~$F$ and $G$ is that in the matrix multiplication, each input entry is used~$k$ times, and in formulas we cannot use intermediate results more than once.)  
Therefore, the recurrence relation $\size(m,n) = 2k\size(m/2, n) + c(k)$ holds. 
Working out the recurrence relation gives $\size(m,n) = (2k)^{\log_2(m)} k^2 w(n) + \bigl[ (2k)^{\log_2(m)} + (2k)^{\log_2(m) -1} + \cdots + 1 \bigr] c(k)$. 
Since $v(n),m(n) \in \polyf(n)$ and since $k$ is constant in~$n$, we have $\size(m(n),n) \in \polyf(n)$. 
We can thus also compute the trace of $M_{n,1} \cdots M_{n,m(n)}$ with a $\polyf(n)$-size formula. 
This shows that $(f_n) \in \VPe$. We thus have $\VPk{k} \subseteq \VPe$. 
\end{proof}

\section{Poly-approximation in width-2 ABPs}\label{sec:cor:interpolation}
We give the interpolation argument that completes the proof of  \protect \MakeUppercase {C}orollary\nobreakspace \ref {cor:interpolation}, which says that the poly-approximation closure of $\VPk{2}$ equals $\VPe$ when $\characteristic(\FF)\neq2$ and $\FF$ is infinite.

\begin{proposition}\label{interp}
$\polyapproxbar{\VPe} = \VPe$ when $\characteristic(\FF) \neq 2$ and $\FF$ is infinite.
\end{proposition}

\begin{proof}
The inclusion $\VPe \subseteq \polyapproxbar{\VPe}$ is clear.
For the other direction, let $(f_n) \in \polyapproxbar{\VPe}$. Then there are polynomials $f_{n;i}(\xvars) \in \FF[\xvars]$, $e(n) \in \polyf(n)$ such that
\[
f_n(\xvars) + \eps f_{n;1}(\xvars) + \eps^2 f_{n;2}(\xvars) + \cdots + \eps^{e(n)} f_{n;e(n)}(\xvars)
\]
is computed by a poly-size formula $\Gamma$ over $\FF(\eps)$. Let $\alpha_0, \alpha_1, \ldots, \alpha_{e(n)}$ be distinct elements in~$\FF$ such that replacing $\eps$ by $\alpha_j$ in $\Gamma$ is a valid substitution (these $\alpha_j$ exist since by assumption our field is infinite). View 
\[
g_n(\eps) \coloneqq f_n(\xvars) + \eps f_{n;1}(\xvars) + \eps^2 f_{n;2}(\xvars) + \cdots + \eps^{e(n)} f_{n;e(n)}(\xvars)
\]
as a polynomial in $\eps$. The polynomial $g_n(\eps)$ has degree at most $e(n)$ so we can write $g_n(\eps)$ as follows (Lagrange interpolation on $e(n) + 1$ points)
\begin{equation}
\label{eq:lagrange}
g_n(\varepsilon) = \sum_{j=0}^{e(n)} g_n(\alpha_j) \prod_{\substack{0 \leq m \leq e(n):\\ m\neq j}} \frac{\varepsilon - \alpha_m}{\alpha_j - \alpha_m}.
\end{equation}
Clearly, $f_n(\xvars) = g_n(0)$. From \eqref{eq:lagrange} we see directly how to write $g_n(0)$ as a linear combination of the values $g_n(\alpha_j)$, namely
\[
g_n(0) = \sum_{j=0}^{e(n)} g_n(\alpha_j) \prod_{\substack{0 \leq m \leq e(n):\\ m\neq j}} \frac{-\alpha_m}{\alpha_j - \alpha_m},
\]
that is,
\[
g_n(0) = \sum_{j=0}^{e(n)}  \beta_j\, g_n(\alpha_j) \quad \textnormal{with} \quad \beta_j \coloneqq \prod_{\substack{0 \leq m \leq e(n):\\ m\neq j}} \frac{\alpha_m}{\alpha_m - \alpha_j}.
\]
The value $g_n(\alpha_j)$ is computed by the formula $\Gamma$ with $\eps$ replaced by $\alpha_j$, which we denote by~$\Gamma|_{\eps=\alpha_j}$. 
Thus $f_n(\xvars)$ is computed by the poly-size formula $\sum_{j=0}^{e(n)} \beta_j\, \Gamma|_{\eps = \alpha_j}$. 
Therefore we have $(f_n) \in \VPe$.
\end{proof}

\begin{remark}\label{interpext}
Proposition\nobreakspace \ref {interp} also holds with $\VPe$ replaced by $\VPs$ or $\VP$ by a similar proof.
\end{remark}

\section{$\VNPk{1}\subsetneq \VNP$ when $\FF = \FF_2$}\label{pro:FIIproof}

In our proofs of $\VNP_1 = \VNP$ (Section\nobreakspace \ref {sec:secondmain} and Section\nobreakspace \ref {sec:VNPdirect}) the assumption  $\characteristic(\FF)\neq 2$ played a crucial role. We can prove that over the finite field $\FF_2$ the inclusion $\VNP_1 \subseteq \VNP$ is indeed strict.

\begin{proposition}\label{pro:FII} 
$\VNPk{1} \subsetneq \VNP$ when $\FF = \FF_2$.
\end{proposition}

\begin{proof}
Let $\FF = \FF_2$.
Clearly $(1+xy) \in \VNP$. 
However, we will prove that $1+xy$ cannot be written as a hypercube sum of affine linear forms. 
In fact, we will prove something stronger, namely that the \emph{function} $(x,y) \mapsto 1 + xy$ cannot be written as a hypercube sum of a product of affine linear forms.

Assume the contrary: the function $(x,y) \mapsto 1 + xy$ can be written as a hypercube sum of a product of affine linear forms. 
We can thus write
\begin{equation}\label{eq:setup}
1+xy = \textstyle\sum_{\bvars} L_{\bvars} \quad \textnormal{with}\quad L_{\bvars} \coloneqq \textstyle\prod_{i=1}^\alpha (x+A_i) \textstyle\prod_{j=1}^\beta (y+B_j) \textstyle\prod_{k=1}^\gamma (x+y+C_k)
\end{equation}
for some affine linear forms $A_i(\bvars)$, $B_j(\bvars)$, $C_k(\bvars)$ in the hypercube variables $\bvars$. On $\FF_2$ the functions $x, x^2, x^3, \ldots$ coincide; the functions $y,y^2,y^3,\ldots$ coincide; and the functions $x+y, (x+y)^2$, $(x+y)^3, \ldots$ coincide, so
\begin{align*}
\textstyle\prod_i (x+A_i) &= \textstyle\prod_i A_i + x \Bigl(\textstyle\prod_i(1+A_i)+\textstyle\prod_i A_i\Bigr),\\
\textstyle\prod_j (y+B_j) &= \textstyle\prod_j B_j + y \Bigl(\textstyle\prod_j(1+B_j)+\textstyle\prod_j B_j\Bigr), \\
\textstyle\prod_k (x+y+C_k) &= \textstyle\prod_k C_k + (x+y) \Bigl(\textstyle\prod_k(1+C_k)+\textstyle\prod_k C_k\Bigr).
\end{align*}
Multiplying the three expressions and simplifying powers of $x$ and $y$ gives
\begin{align*}
L_\bvars &= \textstyle\prod_{i,j,k} A_i B_j C_k + x\Bigl( \textstyle\prod_{i,j,k} (1+A_i)B_j(1+C_k) + \textstyle\prod_{i,j,k}A_i B_j C_k \Bigr)\\
&+y\Bigl( \textstyle\prod_{i,j,k} A_i(1+B_j)(1+C_k) + \textstyle\prod_{i,j,k}A_i B_j C_k \Bigr)\\
&+xy \Bigl( \textstyle\prod_{i,j,k} A_i (1+B_j) (1+C_k) +\textstyle\prod_{i,j,k} (1+A_i) B_j (1+C_k)\\ 
&\qquad + \textstyle\prod_{i,j,k} (1+A_i) (1+B_j) C_k + \textstyle\prod_{i,j,k} A_i B_j C_k\Bigr).
\end{align*}
Plugging in the four possible assignments $(x,y) \in \FF_2\times \FF_2$ into
$1+xy = \sum_{\bvars} L_{\bvars}$,
we get the following system of equations
\begin{align}
\textstyle\sum_\bvars \textstyle\prod_{i,j,k} A_i B_j C_k &= 1,\label{eqline1}\\
\textstyle\sum_\bvars \textstyle\prod_{i,j,k} (1+A_i) B_j (1+C_k) &= 1,\label{eqline2}\\
\textstyle\sum_\bvars \textstyle\prod_{i,j,k} A_i (1+B_j) (1+C_k) &= 1,\label{eqline3}\\
\textstyle\sum_\bvars \textstyle\prod_{i,j,k} (1+A_i) (1+B_j) C_k &= 0.\label{eqline4}
\end{align}

We will show that the above system of equations is inconsistent.
Note that \eqref{eqline1} asserts that an odd number of vectors $\bvars$ satisfy the system of equations
\begin{align*}
A_i &= 1 \ \forall i\\
B_j &= 1 \ \forall j\\
C_k &= 1 \ \forall k.
\end{align*}
Recall that we defined $\alpha, \beta, \gamma$ as the number of factors $x+A_i$, $y+B_j$, $x + y + C_k$ in \eqref{eq:setup}, respectively.
Let $m\coloneqq \alpha + \beta + \gamma$. Recall that we defined $n$ as the number of hypercube variables $b_\ell$. 
As we work over $\FF_2$, any affine linear form in $\bvars$ can be written as $\alpha_0 + \sum_{\ell=1}^n \alpha_\ell b_\ell$ with $\alpha_i \in \{0,1\}$. Write the $i$th linear form in $(A_1, \ldots, A_\alpha, B_1, \ldots, B_\beta, C_1, \dots, C_\gamma)$ as $v_{0,i} + \sum_{\ell=1}^n b_\ell v_{\ell,i}$, and let $v_\ell = (v_{\ell,1}, \dots, v_{\ell,m})$ for $0 \leq \ell \leq n$.
We define the linear map $M: \FF_2^n \to \FF_2^m$ by $M(\bvars) = \sum_{\ell=1}^n b_\ell v_\ell$. We call a bit vector $\bvars \in \FF_2^n$ a \defin{solution} of \eqref{eqline1} if $M(b) = v_0 + 1^\alpha 1^\beta 1^\gamma$, where $1^\alpha 1^\beta 1^\gamma$ is the all-ones vector.
Observe that \eqref{eqline1} says that there is an odd number of solutions of \eqref{eqline1}.
Since the set of solutions of \eqref{eqline1} forms an affine linear subspace of $(\FF_2)^n$, its cardinality is a power of two.
The only odd power of two is $1$, so there is exactly one solution of \eqref{eqline1}.
Let $b^{(1)}$ be this unique solution: $M(b^{(1)})=v_0+1^\alpha 1^\beta 1^\gamma$.
We do the same for \eqref{eqline2} and \eqref{eqline3} and find unique solutions
$M(b^{(2)})=v_0+0^\alpha 1^\beta 0^\gamma$ and
$M(b^{(3)})=v_0+1^\alpha 0^\beta 0^\gamma$.
Equation \eqref{eqline4} asserts that the number of solutions of \eqref{eqline4} is even.
One solution of \eqref{eqline4} is given by $M(b^{(1)}+b^{(2)}+b^{(3)})= 3v_0 + 1^\alpha1^\beta1^\gamma + 0^\alpha 1^\beta 0^\gamma + 1^\alpha 0^\beta 0^\gamma = v_0+0^\alpha 0^\beta 1^\gamma$.
Let $b^{(4')}$ and $b^{(4'')}$ be two distinct solutions of \eqref{eqline4} with
$M(b^{(4')})=M(b^{(4'')})=v_0+0^\alpha 0^\beta 1^\gamma$.
Then $M(b^{(2)}+b^{(3)}+b^{(4')}) = v_0+1^\alpha 1^\beta 1^\gamma = M(b^{(2)}+b^{(3)}+b^{(4'')})$,
which contradicts the uniqueness of $b^{(1)}$.
\end{proof}

\begin{remark}
Our proof of Proposition\nobreakspace \ref {pro:FII} does not generalize to all fields $\FF$ of characteristic~2, because the polynomial $1+xy$ is in fact computable by a hypercube sum of a product of affine linear forms when $\FF=\FF_4$ (and thus when $\FF = \FF_{2^{2k}}$, $k \in \NN$).
Indeed, $\FF_4 \cong \FF_2[Z]/(Z^2+Z+1)$, so the element $Z \in \FF_4$ is a third root of unity ($Z^3=1$)
and satisfies $Z^2+Z+1=0$.
It can be checked that therefore
$\sum_{b=0}^1 (x+Z^2y+Zb)\cdot(x+Zy+Z^2b)\cdot(x+y+b)$ equals $1+xy$.
\end{remark}

\bibliography{smallwidthabp_arxiv_v2.bib}

\appendix
\section{Overview figure}\label{sec:figures}
The diagram in Fig.\nobreakspace \ref {fig:overview} gives an overview of inclusions and separations of complexity classes.

\makeatletter
\tikzset{
  edge node/.code={%
      \expandafter\def\expandafter\tikz@tonodes\expandafter{\tikz@tonodes #1}}}
\makeatother

{

\tikzset{
  equals/.style={
    draw=none,
    edge node={node [sloped, allow upside down, auto=false]{$=$}}},
  subseteq/.style={
    draw=none,
    edge node={node [sloped, allow upside down, auto=false]{$\subseteq$}}},
  subsetneq/.style={
    draw=none,
    edge node={node [sloped, allow upside down, auto=false]{$\subsetneq$}}}
}

\begin{sidewaysfigure}[h!]
\centering
\begin{tikzpicture}
  \matrix (m) [matrix of math nodes,row sep={4em,between origins},column sep={6em,between origins},rectangle]{
	\overline{\VNP_1^\weakest} & \overline{\VNP_1^\weak} & \overline{\VNP_1^\gen} & \overline{\VNP_2^\weakest} & \overline{\VNP_2^\weak} & \overline{\VNP_2^\gen} & \overline{\VNPe} & \overline{\VNPs} & \overline{\VNP} \\
    \VNP_1^\weakest & \VNP_1^\weak & \VNP_1^\gen & \VNP_2^\weakest & \VNP_2^\weak & \VNP_2^\gen & \VNPe & \VNPs & \VNP \\
    \overline{\VP_1^\weakest} & \overline{\VP_1^\weak} & \overline{\VP_1^\gen} & \overline{\VP_2^\weakest} & \overline{\VP_2^\weak} & \overline{\VP_2^\gen} & \overline{\VPe} & \overline{\VPs} & \overline{\VP} \\
    \overline{\VP_1^\weakest}^\poly & \overline{\VP_1^\weak}^{\poly} & \overline{\VP_1^\gen}^\poly & \overline{\VP_2^\weakest}^\poly & \overline{\VP_2^\weak}^\poly & \overline{\VP_2^\gen}^\poly & \overline{\VPe}^\poly & \overline{\VPs}^\poly & \overline{\VP}^\poly \\
    \VP_1^\weakest & \VP_1^\weak & \VP_1^\gen & \VP_2^\weakest & \VP_2^\weak & \VP_2^\gen & \VPe & \VPs & \VP \\
  };
 \path[auto]           (m-1-1) edge[subsetneq] (m-1-2)
                       (m-1-2) edge[subsetneq] (m-1-3)
                       (m-1-3) edge[equals]    (m-1-4)
                       (m-1-4) edge[equals]    (m-1-5)
                       (m-1-5) edge[equals]    (m-1-6)
                       (m-1-6) edge[equals]    (m-1-7)
                       (m-1-7) edge[equals]    (m-1-8)
                       (m-1-8) edge[equals]    (m-1-9)

                       (m-2-1) edge[equals]    (m-1-1)
                       (m-2-2) edge[equals]    (m-1-2) 
                       (m-2-3) edge[subseteq]  (m-1-3) 
                       (m-1-4) edge[subseteq]  (m-2-4) 
                       (m-2-5) edge[subseteq]  (m-1-5) 
                       (m-2-6) edge[subseteq]  (m-1-6)
                       (m-2-7) edge[subseteq]  (m-1-7)
                       (m-2-8) edge[subseteq]  (m-1-8)
                       (m-2-9) edge[subseteq]  (m-1-9)

                       (m-2-1) edge[subsetneq] (m-2-2)
                       (m-2-2) edge[subsetneq] (m-2-3)
                       (m-2-3) edge[equals]    (m-2-4)
                       (m-2-4) edge[equals]    (m-2-5)
                       (m-2-5) edge[equals]    (m-2-6)
                       (m-2-6) edge[equals]    (m-2-7)
                       (m-2-7) edge[equals]    (m-2-8)
                       (m-2-8) edge[equals]    (m-2-9)

                       (m-3-1) edge[equals]    (m-2-1)
                       (m-3-2) edge[equals]    (m-2-2)
                       (m-3-3) edge[subsetneq] (m-2-3)

                       (m-3-1) edge[subsetneq] (m-3-2)
                       (m-3-2) edge[subsetneq] (m-3-3)
                       (m-3-3) edge[subsetneq] (m-3-4)
                       (m-3-4) edge[equals]    (m-3-5)
                       (m-3-5) edge[equals]    (m-3-6)
                       (m-3-6) edge[equals]    (m-3-7)
                       (m-3-7) edge[subseteq]  (m-3-8)
                       (m-3-8) edge[subseteq]  (m-3-9)
                       (m-4-1) edge[equals]    (m-3-1)
                       (m-4-2) edge[equals]    (m-3-2)
                       (m-4-3) edge[equals]    (m-3-3)
                       (m-3-4) edge[subseteq]  (m-4-4) 
                       (m-4-5) edge[subseteq]  (m-3-5) 
                       (m-4-6) edge[subseteq]  (m-3-6) 
                       (m-4-7) edge[subseteq]  (m-3-7)
                       (m-4-8) edge[subseteq]  (m-3-8)
                       (m-4-9) edge[subseteq]  (m-3-9)
                       (m-4-1) edge[subsetneq] (m-4-2)
                       (m-4-2) edge[subsetneq] (m-4-3)
                       (m-4-3) edge[subsetneq] (m-4-4)
                       (m-4-4) edge[equals]    (m-4-5)
                       (m-4-5) edge[equals]    (m-4-6)
                       (m-4-6) edge[equals]    (m-4-7)
                       (m-4-7) edge[subseteq]  (m-4-8)
                       (m-4-8) edge[subseteq]  (m-4-9)
                       (m-5-1) edge[equals]    (m-4-1)
                       (m-5-2) edge[equals]    (m-4-2)
                       (m-5-3) edge[equals]    (m-4-3)
                       (m-5-4) edge[subsetneq] (m-4-4) 
                       (m-5-5) edge[subsetneq] (m-4-5) 
                       (m-5-6) edge[subsetneq] (m-4-6) 
                       (m-5-7) edge[equals]    (m-4-7)
                       (m-5-8) edge[equals]    (m-4-8)
                       (m-5-9) edge[equals]    (m-4-9)
                       (m-5-1) edge[subsetneq] (m-5-2)
                       (m-5-2) edge[subsetneq] (m-5-3)
                       (m-5-3) edge[subsetneq] (m-5-4)
                       (m-5-4) edge[subseteq]  (m-5-5)
                       (m-5-5) edge[subsetneq] (m-5-6)

                       (m-5-6) edge[subsetneq] (m-5-7)
                       (m-5-7) edge[subseteq]  (m-5-8)
                       (m-5-8) edge[subseteq]  (m-5-9)
                       
                       (m-5-1) edge[draw=none] node[yshift=0.4em] {\footnotesize\ref{VP1sep}} (m-5-2)
                       (m-5-2) edge[draw=none] node[yshift=0.4em] {\footnotesize\ref{VP1sep}} (m-5-3)
                       (m-5-3) edge[draw=none] node[yshift=0.4em] {\footnotesize\ref{VP1sep}} (m-5-4)
                       (m-5-5) edge[draw=none] node[yshift=0.4em] {\footnotesize\ref{th:sepweakgen}} (m-5-6)
                       (m-5-6) edge[draw=none] node[yshift=0.4em,rectangle] {\footnotesize\cite{wang}} (m-5-7)

                       (m-1-1) edge[draw=none] node[left,xshift=-0.5em] {\footnotesize\ref{VP1isVNP1weakest}} (m-2-1)
                       (m-2-1) edge[draw=none] node[left,xshift=-0.5em] {\footnotesize\ref{VP1isVNP1weakest}} (m-3-1)
                       (m-3-1) edge[draw=none] node[left,xshift=-0.5em] {\footnotesize\ref{VP1isVP1bar}} (m-4-1)
                       (m-4-1) edge[draw=none] node[left,xshift=-0.5em] {\footnotesize\ref{VP1isVP1bar}} (m-5-1)

                       (m-1-2) edge[draw=none] node[left,xshift=-0.5em] {\footnotesize\ref{VP1isVNP1weakest}} (m-2-2)
                       (m-2-2) edge[draw=none] node[left,xshift=-0.5em] {\footnotesize\ref{VP1isVNP1weakest}} (m-3-2)
                       (m-3-2) edge[draw=none] node[left,xshift=-0.5em] {\footnotesize\ref{VP1isVP1bar}} (m-4-2)
                       (m-4-2) edge[draw=none] node[left,xshift=-0.5em, yshift=1em] {\footnotesize\ref{VP1isVP1bar}} (m-5-2)

                       (m-2-3) edge[draw=none] node[left,xshift=-0.5em] {\footnotesize\ref{VP1isVNP1weakest}} (m-3-3)
                       (m-3-3) edge[draw=none] node[left,xshift=-0.5em] {\footnotesize\ref{VP1isVP1bar}} (m-4-3)
                       (m-4-3) edge[draw=none] node[left,xshift=-0.5em, yshift=1em] {\footnotesize\ref{VP1isVP1bar}} (m-5-3)

                       (m-4-4) edge[draw=none] node[yshift=0.4em] {\footnotesize\ref{moreprecise}} (m-4-5)
                       (m-4-5) edge[draw=none] node[yshift=0.4em] {\footnotesize\ref{moreprecise}} (m-4-6)
                       (m-4-6) edge[draw=none] node[yshift=0.4em] {\footnotesize\ref{approxcor}} (m-4-7)

                       (m-3-4) edge[draw=none] node[yshift=0.4em] {\footnotesize\ref{moreprecise}} (m-3-5)
                       (m-3-5) edge[draw=none] node[yshift=0.4em] {\footnotesize\ref{moreprecise}} (m-3-6)
                       (m-3-6) edge[draw=none] node[yshift=0.4em] {\footnotesize\ref{approxcor}} (m-3-7)
                       
                       (m-2-3) edge[draw=none] node[yshift=0.4em] {\footnotesize\ref{VNPmain}} (m-2-4)
                       (m-2-4) edge[draw=none] node[yshift=0.4em] {\footnotesize\ref{VNPmain}} (m-2-5)
                       (m-2-5) edge[draw=none] node[yshift=0.4em] {\footnotesize\ref{VNPmain}} (m-2-6)
                       (m-2-6) edge[draw=none] node[yshift=0.4em] {\footnotesize\ref{VNPmain}} (m-2-7)
                       (m-2-7) edge[draw=none] node[yshift=0.4em, rectangle] {\footnotesize\cite{valiant1980reducibility}} (m-2-8)
                       (m-2-8) edge[draw=none] node[yshift=0.4em, rectangle] {\footnotesize\cite{valiant1980reducibility}} (m-2-9)
                       
                       (m-5-7) edge[draw=none] node[yshift=0.4em, rectangle] {\footnotesize\cite{Val:79b}} (m-5-8)
                       
                       (m-5-7) edge[draw=none] node[left,xshift=-0.5em, yshift=1em] {\footnotesize\ref{interp}} (m-4-7)
                       (m-5-8) edge[draw=none] node[left,xshift=-0.5em, yshift=1em] {\footnotesize\ref{interpext}} (m-4-8)
                       (m-5-9) edge[draw=none] node[left,xshift=-0.5em, yshift=1em] {\footnotesize\ref{interpext}} (m-4-9)

                       (m-3-4) edge[gray, anchor=center, allow upside down, -latex, shorten >=0pt, shorten <=12pt,bend left=48pt,in=140, edge node={node[sloped, pos=0.1, align=center] {$\subseteq$}}]  (m-1-4)
                       (m-3-5) edge[gray, anchor=center, allow upside down, -latex, shorten >=0pt, shorten <=12pt,bend left=40pt,in=140, edge node={node[sloped, pos=0.1, align=center] {$\subseteq$}}]  (m-1-5)

                       (m-3-6) edge[gray, anchor=center, allow upside down, -latex, shorten >=0pt, shorten <=12pt,bend left=40pt,in=140, edge node={node[sloped, pos=0.1, align=center] {$\subseteq$}}]  (m-1-6)
                       (m-3-7) edge[gray, anchor=center, allow upside down, -latex, shorten >=0pt, shorten <=12pt,bend left=35pt,in=140, edge node={node[sloped, pos=0.1, align=center] {$\subseteq$}}]  (m-1-7) 
                       (m-3-8) edge[gray, anchor=center, allow upside down, -latex, shorten >=0pt, shorten <=12pt,bend left=35pt,in=140, edge node={node[sloped, pos=0.1, align=center] {$\subseteq$}}] (m-1-8)
                       (m-3-9) edge[gray, anchor=center, allow upside down, -latex, shorten >=0pt, shorten <=12pt,bend left=35pt,in=140, edge node={node[sloped, pos=0.1, align=center] {$\subseteq$}}] (m-1-9)

                       (m-4-4) edge[gray, anchor=center, allow upside down, -latex, shorten >=0pt, shorten <=15pt,bend left=38pt, in=130, edge node={node[sloped, pos=0.1, align=center] {$\subseteq$}}]  (m-2-4)
                       (m-4-5) edge[gray, anchor=center, allow upside down, -latex, shorten >=0pt, shorten <=15pt,bend left=35pt, in=130, edge node={node[sloped, pos=0.1, align=center] {$\subseteq$}}]  (m-2-5)

                       (m-4-6) edge[gray, anchor=center, allow upside down, -latex, shorten >=0pt, shorten <=15pt,bend left=35pt, in=130, edge node={node[sloped, pos=0.1, align=center] {$\subseteq$}}]  (m-2-6)
                       (m-4-7) edge[gray, anchor=center, allow upside down, -latex, shorten >=0pt, shorten <=15pt,bend left=35pt, in=140, edge node={node[sloped, pos=0.1, align=center] {$\subseteq$}}]  (m-2-7)
                       (m-4-8) edge[gray, anchor=center, allow upside down, -latex, shorten >=0pt, shorten <=15pt,bend left=35pt, in=140, edge node={node[sloped, pos=0.1, align=center] {$\subseteq$}}] (m-2-8)
                       (m-4-9) edge[gray, anchor=center, allow upside down, -latex, shorten >=0pt, shorten <=15pt,bend left=35pt, in=140, edge node={node[sloped, pos=0.1, align=center] {$\subseteq$}}] (m-2-9);

\end{tikzpicture}
\caption{Overview of inclusions and separations among $\VPk{k}^*$, $\VPe$, $\VPs$, $\VPe$ and their closures when $\characteristic(\FF)\neq 2$.}
\label{fig:overview}
\end{sidewaysfigure}
}



\end{document}